\documentclass[aps,pra,reprint,groupedaddress]{revtex4-2}

\usepackage{multirow}
\usepackage{amssymb,amsmath,amsthm,amsfonts}
\usepackage{bm}
\usepackage{textgreek}
\usepackage{mathtools}
\usepackage{enumitem}
\usepackage{graphicx}
\usepackage{tablefootnote} 
\usepackage{float}
\usepackage[ruled,vlined,linesnumbered]{algorithm2e}
\usepackage{algorithmic}

\usepackage{physics}
\usepackage[normalem]{ulem}

\usepackage{footnote}
\usepackage{xcolor}
\usepackage{mathrsfs}
\usepackage{bbm}
\usepackage{makecell}
\usepackage{pbox}
\usepackage[colorlinks]{hyperref}
\hypersetup{citecolor=blue}
\usepackage{booktabs}
\usepackage{colortbl}
\usepackage{tikz}

\usetikzlibrary{calc,patterns,angles,quotes}
\usepackage{siunitx}
\usetikzlibrary{quantikz}

\newtheorem{theorem}{Theorem}

\newtheorem{lemma}{Lemma}
\newtheorem{proposition}{Proposition}

\newcommand{\eq}[1]{Eq.~(\ref{eq:#1})}
\newcommand{\thm}[1]{\hyperref[thm:#1]{Theorem~\ref*{thm:#1}}}
\newcommand{\cor}[1]{\hyperref[cor:#1]{Corollary~\ref*{cor:#1}}}
\newcommand{\defn}[1]{\hyperref[defn:#1]{Definition~\ref*{defn:#1}}}
\newcommand{\lem}[1]{\hyperref[lem:#1]{Lemma~\ref*{lem:#1}}}
\newcommand{\prop}[1]{\hyperref[prop:#1]{Proposition~\ref*{prop:#1}}}
\newcommand{\fig}[1]{\hyperref[fig:#1]{Figure~\ref*{fig:#1}}}
\newcommand{\tab}[1]{\hyperref[tab:#1]{Table~\ref*{tab:#1}}}
\newcommand{\algo}[1]{\hyperref[algo:#1]{Algorithm~\ref*{algo:#1}}}
\renewcommand{\sec}[1]{\hyperref[sec:#1]{Section~\ref*{sec:#1}}}
\newcommand{\append}[1]{\hyperref[append:#1]{Appendix~\ref*{append:#1}}}
\newcommand{\fac}[1]{\hyperref[fac:#1]{Fact~\ref*{fac:#1}}}
\newcommand{\lin}[1]{\hyperref[lin:#1]{Line~\ref*{lin:#1}}}

\renewcommand{\i}{\mathrm{i}}
\newcommand{\eps}{\varepsilon}

\def\>{\rangle}
\def\<{\langle}

\renewcommand{\bra}[1]{\langle#1|}
\renewcommand{\ket}[1]{|#1\rangle}

\newcommand{\E}{\mathbb{E}}

\newcommand{\mc}{\mathcal}

\DeclareMathOperator{\poly}{poly}

\DeclareMathOperator{\supp}{supp}
\DeclareMathOperator{\ad}{ad}
\DeclareMathOperator{\sgn}{sgn}
\DeclareMathOperator{\CPTP}{CPTP}

\allowdisplaybreaks

\renewcommand{\emptyset}{\varnothing}
\def\:{\hbox{\bf:}}

\makeatletter
\newcommand{\doubletilde}[1]{{%
  \mathpalette\double@tilde{#1}%
}}
\newcommand{\double@tilde}[2]{%
  \sbox\z@{$\m@th#1\tilde{#2}$}%
  \ht\z@=.9\ht\z@
  \tilde{\box\z@}%
}
\makeatother

\let\oldnl\nl
\newcommand{\nonl}{\renewcommand{\nl}{\let\nl\oldnl}}

\SetKwComment{Comment}{/* }{ */}
\SetKwInput{KwInput}{Input}                
\SetKwInput{KwOutput}{Output}              
\SetKwInput{KwNotation}{Notation}              
\SetKwInput{KwParameter}{Parameter}             
\SetKwIF{WithProb}{ElseWith}{Else}{with probability}{}{else}{otherwise}{endif}

\begin{document}

\title{Trotter error compensation with polylogarithmic precision and nested-commutator scaling without ancillas}

\author{Xinzhao Wang$^{1,2}$}
\author{Shuo Zhou$^{1,2}$}
\author{Ziruo Wang$^{1,2}$}

\author{Pei Zeng$^{3,4}$}

\author{Jinzhao Sun$^{5}$}

\author{Qi Zhao$^{6}$}

\author{Tom Gur$^{7}$}
\thanks{Corresponding author. Email: \href{mailto:tom.gur@cl.cam.ac.uk}{tom.gur@cl.cam.ac.uk}}

\author{Tongyang Li$^{1,2}$}
\thanks{Corresponding author. Email: \href{mailto:tongyangli@pku.edu.cn}{tongyangli@pku.edu.cn}}

\affiliation{\vspace{1em}\textsuperscript{1}Center on Frontiers of Computing Studies, Peking University, Beijing 100871, China}

\affiliation{\textsuperscript{2}School of Computer Science, Peking University, Beijing 100871, China}

\affiliation{\textsuperscript{3}Institute of Natural Sciences, Shanghai Jiao Tong University, Shanghai 200240, China}

\affiliation{\textsuperscript{4}School of Physics and Astronomy, Shanghai Jiao Tong University, Shanghai 200240, China}

\affiliation{\textsuperscript{5}School of Physical and Chemical Sciences, Queen Mary University of London, London E1 4NS, United Kingdom}

\affiliation{\textsuperscript{6}QICI Quantum Information and Computation Initiative, School of Computing and Data Science,
The University of Hong Kong, Pokfulam Road, Hong Kong SAR, China}

\affiliation{\textsuperscript{7}Department of Computer Science and Technology, University of Cambridge, Cambridge CB3 0FD, UK}
\date{\today}

\begin{abstract}
Product formulas are among the most practical approaches to Hamiltonian
simulation, requiring no ancillary qubits and exhibiting error bounds governed
by nested commutators rather than only by Hamiltonian norms. Their circuit
size, however, scales polynomially with the inverse precision. We develop a
high-order nested-commutator compensation (HNCC) algorithm that preserves the
main advantages of product formulas while achieving polylogarithmic precision
dependence in the circuit size and the standard
\(\mathcal{O}(\varepsilon^{-2})\) sampling cost. HNCC uses a truncated
Baker--Campbell--Hausdorff expansion to represent high-order Trotter errors
by products of nested commutators and compensates these errors at the
channel level through randomly sampled Pauli-rotation channels, avoiding
Hadamard tests and ancillary qubits. For a fixed \(K\)-th order product formula
applied to a \(k\)-local Hamiltonian on \(N\) qubits with \(\Gamma\)
Pauli terms and local interaction strength \(g_0\), HNCC estimates
\(\operatorname{tr}[Oe^{-\i tH}\rho e^{\i tH}]\) to additive precision
\(\varepsilon\|O\|\) using \(\mathcal{O}(\varepsilon^{-2})\) repetitions. Its
maximum gate count per circuit is
\(
\mathcal{O}\bigl(
    kN^{\frac{1}{2K+1}}
    \Gamma^{1-\frac{1}{2K+1}}
    \max\{\Gamma,N\log(1/\varepsilon)\}^{\frac{1}{2K+1}}
    (kg_0t\log(1/\varepsilon))^{1+\frac{1}{2K+1}}
\bigr)\).
Finite-size resource estimates for the periodic Heisenberg chain indicate
that HNCC has the lowest estimated \(T\)-gate count per circuit among the
product-formula-based methods considered.
\end{abstract}

\maketitle

\section{Introduction}
Hamiltonian simulation, the task of simulating the dynamics generated by a
Hamiltonian, is a fundamental problem in quantum
computing~\cite{feynman1982simulating,lloyd1996universal}. Beyond its direct
application to studying the dynamics of physical
systems~\cite{mcardle2020quantum,childs2018toward}, it serves as an
algorithmic primitive for various computational tasks, including ground-state
preparation~\cite{dong2022ground,sun2026high}, solving linear systems of
equations~\cite{harrow2009quantum,childs2017quantum} and differential
equations~\cite{an2023linear}, and combinatorial
optimization~\cite{farhi2014quantum}.

\citet{lloyd1996universal} introduced the first explicit quantum algorithm
for this task based on product formulas~\cite{suzuki1991general,
childs2021theory}, also known as the Trotter method. For a Hamiltonian
\(H=\sum_{\gamma=1}^{\Gamma}H_\gamma\), a \(K\)-th order product formula
\(S_K(x)\), composed of evolutions under the individual terms \(H_\gamma\),
approximates \(U(x)=e^{-\i Hx}\) to \(K\)-th order in \(x\). Although various
``post-Trotter methods'' were subsequently
developed~\cite{berry2014exponential,berry2015simulating,
berry2015hamiltonian,low2017optimal,low2019hamiltonian}, product formulas
remain practically appealing because of their simple circuit structure, low
ancilla requirements, and low gate count~\cite{childs2021theory}.
Furthermore, their error bounds can exploit nested commutators among the terms
\(\{H_\gamma\}\), rather than depending only on the norm sum
\(\sum_{\gamma=1}^{\Gamma}\|H_\gamma\|\)~\cite{childs2021theory}. For local
lattice systems with \(N\) sites, this commutator scaling can yield a
system-size dependence of \(\mathcal{O}(N^{1+1/K})\) for a \(K\)-th order
product formula, compared with the \(\mathcal{O}(N^2)\) dependence of
post-Trotter methods implemented using the standard LCU block encoding in the
same setting.

Despite these advantages, product formulas exhibit a suboptimal dependence on the simulation time $t$ and the target precision $\eps$. The gate complexity of a $K$-th order method scales as $\mathcal{O}(t^{1+1/K} \eps^{-1/K})$, in contrast to the optimal linear time dependence and logarithmic precision dependence. For standard Suzuki--Trotter formulas, the implementation cost of one step grows exponentially with the order $K$, making low orders such as $K=2$ and $K=4$ more favorable in practice. In this regime, the polynomial dependence on both $t$ and $1/\eps$ becomes a bottleneck for long-time or high-precision simulation.

For observable estimation, recent works have developed two main approaches to reducing Trotter errors in $\tr[O U(t)\rho_0 U(t)^{\dagger}]$: classical post-processing~\cite{endo2019mitigating,carrera2023wellconditioned,rendon2024improved,watson2025exponentially} and Trotter error compensation~\cite{cho2024doubling,zeng2025simple,mineh2025improving}. Richardson extrapolation
combines expectation values measured from a \(K\)-th order Trotter circuit at
several step sizes and extrapolates to the zero-step-size
limit~\cite{watson2025exponentially,mizuta2026commutator}. This achieves
polylogarithmic precision dependence in the circuit size while preserving
commutator scaling and requiring no ancillas. These advantages come with an increased sampling cost because the extrapolation
coefficients amplify statistical noise, as well as SPAM errors that are difficult
to suppress experimentally. For the extrapolation formulas considered in
Refs.~\cite{watson2025exponentially,mizuta2026commutator}, the repetition
complexity is
\(\mathcal{O}(\eps^{-2}\log(1/\eps)(\log\log(1/\eps))^3)\), rather than the
standard \(\mathcal{O}(\eps^{-2})\). The simulation-time dependence remains
that of the original \(K\)-th order product formula,
\(\mathcal{O}(t^{1+1/K})\).

Another representative approach, proposed by \citet{zeng2025simple},
compensates Trotter errors by applying unitaries sampled from a representation
of the Trotter remainder after each Trotter step. Their nested-commutator
compensation (NCC) method preserves commutator scaling and achieves the same
simulation-time dependence as a product formula of order \(2K+1\). Its
currently proved precision dependence is polynomial in \(1/\eps\), although it is
conjectured to be polylogarithmic~\cite{sun2026high}.
Their paired Taylor-series compensation (PSTC) method instead achieves
polylogarithmic precision dependence by representing higher-order errors
through products \(H_{\gamma_1}\cdots H_{\gamma_m}\), trading
nested-commutator scaling for improved precision dependence. Both methods
perform the compensation at the unitary level, producing cross-terms of the
form \(\tr\big[O V_i\rho_0V_j^\dagger\big]\). Their standard implementation estimates
these terms using a Hadamard test and requires an ancillary qubit and
controlled compensation operations, whereas the ancilla may not be needed for symmetry-conserved systems~\cite{sun2026high}. On devices with limited connectivity,
such controlled operations can reduce circuit parallelism and increase circuit
depth.

In this work, we develop a high-order nested-commutator compensation (HNCC)
algorithm that combines commutator scaling, effective \(2K+1\)-order time
scaling, polylogarithmic precision dependence in the maximum gate count per
circuit, and an ancilla-free implementation. Let \(\mc U(x)\) and
\(\mc S_K(x)\) denote the unitary channels induced by \(U(x)\) and \(S_K(x)\),
respectively. We define the multiplicative Trotter remainder by
\(\mc V_K(x)=\mc U(x)\mc S_K^{\dagger}(x)\), so that applying
\(\mc V_K(x)\) after \(\mc S_K(x)\) recovers the ideal evolution
\(\mc U(x)\). Our main idea is to approximate \(\mc V_K(x)\) to high order
by a \emph{linear combination of quantum channels (LCQC)}, each implemented
as a composition of Pauli-rotation channels. Applying a randomized implementation of this LCQC after \(\mc S_K(x)\)
approximately recovers \(\mc U(x)\) in expectation.
Since the correction is applied at the channel level, it avoids the
cross-terms arising in unitary-level compensation and hence the need for
Hadamard tests or ancillary qubits.

The main technical challenge is to construct such an LCQC approximation
to arbitrarily high order while preserving the nested-commutator structure
of the Trotter error. The Baker--Campbell--Hausdorff (BCH) expansion expresses
the generator of \(\mc V_K(x)\) in terms of nested commutators.
Using BCH truncation-error bounds based on doubly right-nested
commutators~\cite{wang2026lindbladian}, we retain only finitely many BCH terms
and therefore do not rely on convergence of the full BCH series, which cannot
generally be assumed for local many-body systems~\cite{mizuta2026commutator}. After expanding the Hamiltonian terms in the Pauli basis, we rewrite the
truncated generator as a linear combination of adjoint-Pauli
superoperators \(\ad_{\i P}\) and expand its exponential. We control the
resulting series by retaining only samples for which the total BCH order of
the higher-order Taylor terms over the full circuit is below a chosen
threshold.
A parameter-shift identity converts each
adjoint-Pauli factor into a difference of Pauli-rotation channels, thereby
yielding an LCQC approximation to the Trotter remainder \(\mc V_K(x)\). For the
part linear in the truncated BCH generator, we instead pair each sampled
adjoint-Pauli term with the identity channel. This reduces the coefficient
\(1\)-norm of the LCQC and hence the sampling overhead, leading to the
effective \(2K+1\)-order time scaling while preserving commutator scaling.

We also provide finite-size resource estimates for the periodic Heisenberg
chain, comparing HNCC with other product-formula-based methods. When applied
to the second-order product formula \(S_2\), HNCC achieves up to a
\(25.9\times\) reduction in the estimated \(T\)-gate count per circuit
relative to the uncompensated \(S_2\) formula.

\tab{feature-comparison} summarizes the qualitative comparison between HNCC and product-formula-based approaches.

\begin{table*}[htpb]
    \centering
    \caption{Qualitative comparison of product-formula-based Hamiltonian simulation approaches.}
    \label{tab:feature-comparison}
    \begin{tabular}{lccccc}
    \toprule
    \textbf{Feature} & \textbf{Trotter~\cite{suzuki1991general}} & \makecell{\textbf{Extrapolation}~\cite{watson2025exponentially,mizuta2026commutator}} & \textbf{PSTC}~\cite{zeng2025simple} & \textbf{NCC}~\cite{zeng2025simple} & \makecell{\textbf{HNCC}} \\
    \midrule
    Commutator scaling & $\checkmark$ & $\checkmark$ & $\times$ & $\checkmark$ & $\checkmark$ \\
    Effective $2K+1$-order time scaling & $\times$ & $\times$ & $\checkmark$ & $\checkmark$ & $\checkmark$ \\
    \makecell[l]{Polylogarithmic precision dependence \\ in circuit size} & $\times$ & $\checkmark$ & $\checkmark$ & $\times$ & $\checkmark$ \\
    $\mc O(\eps^{-2})$ circuit repetitions & $\checkmark$ & $\times$ & $\checkmark$ & $\checkmark$ & $\checkmark$ \\
    Ancilla-free implementation & $\checkmark$ & $\checkmark$ & $\times$ & $\times$ & $\checkmark$ \\
    \bottomrule
    \end{tabular}
\end{table*}

\section{Summary of results}
We summarize the complexity bounds for HNCC and the LCQC sampling and BCH
truncation used to obtain them.

Let $H$ be a $k$-local Hamiltonian on a set of sites $\Lambda=[N]$, with one
qubit on each site, written as
\[
    H=\sum_{\gamma=1}^{\Gamma}\alpha_\gamma P_\gamma,
\]
where $\alpha_\gamma\in\mathbb R$ and each $P_\gamma$ is an $N$-qubit Pauli
operator supported on at most $k$ sites. We assume that the local interaction
strength of $H$ is bounded by $g_0$, as formalized by the
$g_0$-extensiveness condition in \eq{def-extensive}. Let
$\mc U(t):=e^{-\i t\ad_H}$ denote the time-evolution channel generated by
$H$, so that $\mc U(t)(\rho)=e^{-\i tH}\rho e^{\i tH}$. The task is to estimate
\[
\tr[O\mc U(t)(\rho_0)]
\]
for a given initial state $\rho_0$ and observable $O$.

\subsection{Main result}

The following theorem states the complexity of HNCC. We use CNOT and arbitrary
single-qubit gates as the elementary gate set in these bounds.

\begin{theorem}[HNCC complexity]
\label{thm:fixed-complexity}
Let $H=\sum_{\gamma=1}^{\Gamma}\alpha_\gamma P_\gamma$ be a $k$-local
Hamiltonian on $N$ qubits satisfying the $g_0$-extensiveness condition. Let
$t>0$ and $0<\eps<1/2$, and suppose that
$\Gamma=\Omega(\log(1/\eps))$ and $kg_0t\ge1$. For any fixed Trotter order
$K$, any initial state $\rho_0$, and any observable $O$, HNCC estimates
$\tr[O\mc U(t)(\rho_0)]$ to additive precision $\eps\|O\|$ with constant
success probability using $\mathcal O(\eps^{-2})$ independent circuit
repetitions. The maximum gate count per circuit is
\begin{multline*}
        \mathcal O\Bigl(
        kN^{\frac{1}{2K+1}}
        \Gamma^{1-\frac{1}{2K+1}}
        \max\bigl\{
            \Gamma,\,
            N\log(1/\eps)
        \bigr\}^{\frac{1}{2K+1}}
        \\
        {}
        \bigl(kg_0t\log(1/\eps)\bigr)^{1+\frac{1}{2K+1}}
    \Bigr).
\end{multline*}
\end{theorem}

\begin{table*}[htpb]
    \centering
    \small
    \caption{Quantitative comparison of Hamiltonian simulation algorithms
    for $k$-local Hamiltonians with local interaction strength bounded by
    $g_0$. We assume $k g_0 t\ge1$ and take $k$ and $K$ to be fixed. The
    last five methods mitigate Trotter errors in a $K$-th order product formula.}
    \label{tab:complexity-comparison}
    \begin{tabular}{llll}
    \toprule
    \textbf{Method} & \makecell{\textbf{Maximum gate count}\\\textbf{per circuit}} & \textbf{Repetitions} & \textbf{Ancillas} \\ \midrule
    Trotter~\cite{suzuki1991general,childs2021theory}
    & $\mc O\big(\Gamma N^{\frac{1}{K}}(g_0t)^{1+\frac{1}{K}}\eps^{-\frac{1}{K}}\big)$
    & $\mc O(\eps^{-2})$ & 0 \\
    Qubitization~\cite{low2019hamiltonian}
    & $\mc O\big(\Gamma(Ng_0t+\log\frac{1}{\eps})\big)$
    & $\mc O(\eps^{-2})$ & $\mc O(\log\Gamma)$ \\ \midrule
    Extrapolation~\cite{watson2025exponentially,mizuta2026commutator}
    & $\mc O\big(\Gamma N^{\frac{1}{K}}
    (g_0t\log\frac{1}{\eps})^{1+\frac{1}{K}}\log\frac{1}{\eps}\big)$
    & $\mc O\big(\eps^{-2}\log\frac{1}{\eps}(\log\log\frac{1}{\eps})^3\big)$
    & 0 \\
    PSTC~\cite{zeng2025simple}
    & $\mc O\big((\Gamma+\log\frac{1}{\eps})
    N^{1+\frac{1}{2K+1}}(g_0t)^{1+\frac{1}{2K+1}}\big)$
    & $\mc O(\eps^{-2})$ & 1 \\
    NCC~\cite{zeng2025simple}
    & $\mc O\big(\Gamma N^{\frac{2}{2K+1}}
    (g_0t)^{1+\frac{1}{2K+1}}\eps^{-\frac{1}{2K+1}}\big)$
    & $\mc O(\eps^{-2})$ & 1 \\
    HNCC (\thm{fixed-complexity})
    & \makecell[l]{$\mc O\bigl(
    \Gamma^{1-\frac{1}{2K+1}}
    N^{\frac{1}{2K+1}}
    \max\{\Gamma,N\log\frac{1}{\eps}\}^{\frac{1}{2K+1}}
    (g_0t\log\frac{1}{\eps})^{1+\frac{1}{2K+1}}
    \bigr)$}
    & \textbf{$\mc O(\eps^{-2})$} & 0 \\
    Unpaired HNCC (\thm{unpaired-global-order})
    & $\mc O\bigl(
    \Gamma N^{\frac{1}{K}}
    (g_0t\log\frac{1}{\eps})^{1+\frac{1}{K}}
    \bigr)$
    & $\mc O(\eps^{-2})$ & 0 \\
    \bottomrule
    \end{tabular}
\end{table*}

The quantitative comparison with representative methods is summarized in
\tab{complexity-comparison}. For Hamiltonians with $\Gamma=\Omega(N)$ Pauli
terms, the general HNCC bound implies the following simpler upper bound on
the maximum gate count per circuit:
\begin{align*}
    \mathcal O\Bigl(
        \Gamma
        N^{\frac{1}{2K+1}}
        (g_0t)^{1+\frac{1}{2K+1}}
        \log^{1+\frac{2}{2K+1}}\frac{1}{\eps}
    \Bigr),
\end{align*}
which improves the system-size and time dependence of the $K$-th order
product formula to that of a product formula of order $2K+1$, while achieving
polylogarithmic precision dependence without ancillary qubits.

We also analyze an unpaired variant of HNCC. Its maximum gate count per circuit
has the same system-size and simulation-time dependence as the original
$K$-th order product formula while achieving polylogarithmic precision
dependence. Although it does not attain the improved system-size and time
scaling of HNCC, it requires only $\mathcal O(k\log(1/\eps))$ additional
Clifford gates beyond the product-formula circuit and no additional $T$ gates.

\subsection{LCQC sampling}
\label{sec:lcqc-framework}

We implement the evolution over total time $t$ by dividing the interval into
$\nu$ steps of size $x=t/\nu$. We partition the Hamiltonian into $L$ groups
of commuting terms, $H=\sum_{\ell=1}^{L}H_\ell$. For ordered products, we use
$\prod_{r=1}^{j}A_r:=A_j\cdots A_1$ and
$\prod_{r=j}^{1}A_r:=A_1\cdots A_j$. We denote the channel
associated with a $K$-th order product-formula step by
\[
    \mc S_K(x)(\rho)=S_K(x)\rho S_K^\dagger(x),
    \qquad
    S_K(x)=\prod_{v=1}^{\kappa_KL}e^{-\i a_vxH_{\ell_v}}.
\]
Here, $\kappa_K$ is the number of product-formula stages per group,
$a_v\in\mathbb R$ are the product-formula coefficients, and
$\ell_v\in[L]$ specifies the ordering of the groups. The $K$-th order
condition gives $\mc S_K(x)=\mc U(x)+\mathcal O(x^{K+1})$.

Ideally, for a single Trotter step, one could compensate the Trotter error by
applying the Trotter remainder
$\mc V_K(x)=\mc U(x)\mc S_K^\dagger(x)$ after the product formula
$\mc S_K(x)$, thereby recovering the exact short-time evolution $\mc U(x)$.
Suppose that $\mc V_K(x)$ admits a representation as a linear combination of
quantum channels,
\begin{align}
\label{eq:lcqc-V}
    \mc V_K(x)=\sum_j\beta_j\mc N_j,
\end{align}
where $\beta_j\in\mathbb R$ and each $\mc N_j$ is either an efficiently
implementable quantum channel or the zero map. In HNCC, the nonzero channels
are compositions of Pauli rotations. Let
\begin{align}
\label{eq:lcqc-1-norm}
    \lambda=\sum_j|\beta_j|
\end{align}
be the \emph{LCQC $1$-norm}. Then \eq{lcqc-V} implies that
$\mc V_K(x)/\lambda$ can be realized in expectation as a randomized signed
channel $\eta\mc E$:
\[
    \mathbb E[\eta\mc E]=\frac{\mc V_K(x)}{\lambda},
\]
where
\[
    \Pr\bigl[
        (\eta,\mc E)=
        \bigl(\sgn(\beta_j),\mc N_j\bigr)
    \bigr]
    =
    \frac{|\beta_j|}{\lambda}.
\]

For the total evolution over $\nu$ Trotter steps, we independently sample a
pair $(\eta_k,\mc E_k)$ for each step $k\in[\nu]$, apply the channels
$\mc E_k\mc S_K(x)$ sequentially, and record the signs $\eta_k$ classically.
Measuring the observable $O$ and multiplying the outcome by
$\prod_{k=1}^{\nu}\eta_k$ yields an estimator whose expectation is
\begin{align*}
    &\mathbb E\biggl[
        \biggl(\prod_{k=1}^{\nu}\eta_k\biggr)
        \tr\bigl[
            O\mc E_\nu\mc S_K(x)\cdots
            \mc E_1\mc S_K(x)(\rho_0)
        \bigr]
    \biggr]
    \\
    ={}&
    \tr\biggl[
        O
        \prod_{k=1}^{\nu}
        \mathbb E[\eta_k\mc E_k\mc S_K(x)]
        (\rho_0)
    \biggr]
    \\
    ={}&
    \tr\biggl[
        O\Bigl(\frac{\mc U(x)}{\lambda}\Bigr)^\nu
        (\rho_0)
    \biggr]
    =
    \frac{1}{\lambda^\nu}
    \tr[O\mc U(t)(\rho_0)].
\end{align*}

Since each measurement outcome is bounded by $\|O\|$, Hoeffding's inequality
implies that estimating the target expectation requires
$\mathcal O(\lambda^{2\nu}/\eps^2)$ independent repetitions. Therefore,
reducing the sampling overhead amounts to constructing an LCQC representation
of the Trotter remainder with small $1$-norm. Here, $\lambda$ is the LCQC
$1$-norm of the complete compensation map applied after one product-formula
step. If this map is implemented by several independently sampled LCQCs, their
$1$-norms multiply.

To obtain such a representation, we truncate the BCH expansion of the Trotter
remainder $\mc V_K(x)$ and sample the retained nested-commutator terms with the
light-cone sampler described below.
Section~\ref{sec:trotter-lcqc} gives the resulting LCQC representation, its
$1$-norm bound, and an efficient method for sampling from it.
Section~\ref{sec:complexity-analysis} derives the overall complexity.

\subsection{Linear combination of quantum channels for the Trotter remainder}
\label{sec:trotter-lcqc}

To obtain higher-order compensation terms while retaining
nested-commutator structure, we use the BCH expansion. For operators $A$ and
$B$, we write $\ad_A(B)=[A,B]$. The leading nonzero term in the BCH expansion
of the Trotter remainder already reveals its commutator structure. For
instance, for the first-order product formula
$S_1(x)=e^{-\i xH_1}e^{-\i xH_2}$,
\begin{align*}
    \mc V_1(x)
    &=
    e^{-\i x\ad_{H_1+H_2}}
    e^{\i x\ad_{H_2}}
    e^{\i x\ad_{H_1}}
    \\
    &=
    \exp\Bigl(
        \frac{x^2}{2}
        [\ad_{H_1},\ad_{H_2}]
        +\mathcal O(x^3)
    \Bigr).
\end{align*}
The notation $\mathcal O(x^3)$ hides the higher-order terms in the BCH series
in the exponent. The displayed commutator explains where commutator scaling
comes from, but by itself it does not bound the full Trotter error. Indeed, for
many physically relevant Hamiltonians, the BCH series in the exponent can
diverge~\cite{mizuta2026commutator}. This is why commutator-scaling analyses
of product formulas often use error representations that do not rely on the
full BCH series~\cite{childs2021theory}.

The NCC algorithm~\cite{zeng2025simple} uses a BCH-free commutator
representation to compensate the leading Trotter error while preserving
commutator scaling. However, this representation does not readily extend to
the higher-order compensation terms required for polylogarithmic precision
dependence. Sun et al. conjectured that one may simultaneously achieve
nested-commutator scaling and polylogarithmic precision in lattice
models~\cite{sun2026high}. The advantage of the BCH expansion is that it gives
nested-commutator representations of compensation terms at arbitrary order.
We retain finitely many BCH terms and bound the truncation error using
Refs.~\cite{mizuta2026commutator,wang2026lindbladian}.

In general, the Trotter remainder can be written as a product of
exponentials:
\begin{align}
    \mc V_K(x)
    &=
    \mc U(x)\mc S_K^\dagger(x)
    \nonumber\\
    &=
    e^{-\i x\ad_H}
    \prod_{v=\kappa_KL}^{1}
    e^{\i a_vx\ad_{H_{\ell_v}}}.
\end{align}
Formally, the BCH formula gives
\[
    \mc V_K(x)
    =
    \exp\biggl(
        \sum_{q=K+1}^{\infty}\Phi_q(x)
    \biggr),
\]
where $\Phi_q(x)$ is the BCH term of order $q$. It is proportional to $x^q$
and is a linear combination of $q$-fold right-nested commutators of the generators
$\{\i x\ad_{H_\ell}\}_{\ell=1}^{L}$ and $-\i x\ad_H$.

By \lem{bound-Phi-q},
\[
    \|\Phi_q(x)\|
    \le
    \frac{q!(2k\tilde g)^qk^{-1}N}{q^2},
    \qquad
    \tilde g=
    (\max_v|a_v|\kappa_K+1)g_0x.
\]
For any fixed $x>0$, the factorial growth makes the sum of these upper bounds
diverge.
We therefore truncate the series at order $q_0$ and define the truncated remainder
\begin{align}
\label{eq:exp-remainder}
    \tilde{\mc V}_{K}^{(q_0)}(x)
    =
    \exp\biggl(
        \sum_{q=K+1}^{q_0}\Phi_q(x)
    \biggr).
\end{align}
The truncation error can be controlled via a doubly right-nested commutator bound~\cite{wang2026lindbladian}, as shown in \lem{truncation-BCH}.

\subsubsection{Sampling for a BCH term}
Before constructing the LCQC for the truncated remainder
$\tilde{\mc V}_{K}^{(q_0)}(x)$, we first give an LCQC for $\Phi_q(x)$. By
decomposing each $H_{\ell}$ and $H$ into the Pauli operators
$\{P_{\gamma}\}$, each BCH term $\Phi_q(x)$ expands into a linear combination
of $q$-fold nested commutators of the adjoint Pauli actions
$\{\ad_{\i P_{\gamma}}\}$. By recursively applying the Jacobi identity,
$[\ad_A,\ad_B]=\ad_{[A,B]}$, each nested commutator simplifies to
\begin{align*}
    &x^q[\ad_{\i P_{\gamma_1}},
    [\ad_{\i P_{\gamma_2}},[\cdots[
    \ad_{\i P_{\gamma_{q-1}}},\ad_{\i P_{\gamma_q}}]\cdots]]]
    \\
    ={}&\!
    x^q\ad_{[\i P_{\gamma_1},
    [\i P_{\gamma_2},[\cdots[
    \i P_{\gamma_{q-1}},\i P_{\gamma_q}]\cdots]]]}
    \in
    \bigl\{
        0,\pm2^{q-1}x^q\ad_{\i P}
    \bigr\},
\end{align*}
for some Pauli operator $P$. The second line follows from the closure of
imaginary Pauli operators under commutation, with a factor of $2$ arising at
each nonzero nesting level. Let $\mc U_V(\rho)=V\rho V^\dagger$. The
parameter-shift rule writes $\ad_{\i P}$ as a difference of two unitary
channels:
\begin{align}
\label{eq:para-shift}
    \ad_{\i P}(\rho)
    ={}&
    (\i P\rho-\i\rho P)
    \nonumber\\
    ={}&
    \frac{1}{2}\bigl[
        (I+\i P)\rho(I-\i P)
        -
        (I-\i P)\rho(I+\i P)
    \bigr]
    \nonumber\\
    ={}&
    \Bigl(
        \mc U_{e^{\i\frac{\pi}{4}P}}
        -
        \mc U_{e^{-\i\frac{\pi}{4}P}}
    \Bigr)(\rho).
\end{align}
Thus each BCH term $\Phi_q(x)$ is a linear combination of Pauli-rotation
channels. Sampling this LCQC reduces to sampling a nested commutator
\[
    [\i\alpha_{v_q}P_{v_q},
    [\i\alpha_{v_{q-1}}P_{v_{q-1}},[\cdots[
    \i\alpha_{v_2}P_{v_2},\i\alpha_{v_1}P_{v_1}]\cdots]]]
\]
with probability proportional to its norm, given a set of scaled Pauli
operators $\{\alpha_vP_v\}_{v=1}^{M}$, where each $P_v$ has support size at
most $k$. Directly enumerating the exact distribution over all index tuples
$(v_1,\ldots,v_q)$ costs $\mathcal O(M^q)$. This is polynomial in $M$ only
for fixed $q$. Since $q$ can be as large as $q_0$, the cost grows
exponentially with the BCH truncation order. Crucially, the Pauli algebraic
structure ensures that any two Pauli operators either commute or
anti-commute, which simplifies the norm of any nonzero nested commutator to
the product of the individual coefficients:
\begin{align*}
    &\|[\i\alpha_{v_q}P_{v_q},
    [\i\alpha_{v_{q-1}}P_{v_{q-1}},[\cdots[
    \i\alpha_{v_2}P_{v_2},\i\alpha_{v_1}P_{v_1}]\cdots]]]\|
    \\
    ={}&
    2^{q-1}\prod_{j=1}^{q}|\alpha_{v_j}|.
\end{align*}
Sampling each index $v_j$ independently with probability proportional to
$|\alpha_{v_j}|$ therefore gives the correct weights. However, many sampled
tuples yield a zero nested commutator, for example when a newly sampled term
has disjoint support from the preceding nested commutator. The resulting LCQC
$1$-norm scales with $(\sum_v|\alpha_v|)^q$, which can scale as $N^q$ for an
extensive local Hamiltonian, far exceeding the $\mathcal O(N)$ commutator
bound for $\|\Phi_q\|$.

We instead use a light-cone sampler similar to that of
Ref.~\cite{zeng2025simple}. Instead of independent sampling, we sequentially
sample $v_j$ for $j$ from $1$ to $q$ while tracking the union of the supports
of $P_{v_1},\ldots,P_{v_j}$. For the next operator $P_{v_{j+1}}$, we restrict
the sampling space to those terms that overlap with this union support. Since
$P_{v_{j+1}}$ commutes with
\(
    [P_{v_j},[P_{v_{j-1}},[\cdots[P_{v_2},P_{v_1}]\cdots]]]
\)
when their supports are disjoint, this restriction eliminates tuples that
vanish because of disjoint support. The resulting LCQC for $\Phi_q(x)$ has
$1$-norm
\begin{align}
\label{eq:lambda-q-summary}
    \lambda_q
    =
    \frac{q!(2k\tilde g)^qk^{-1}N}{q^2},
\end{align}
which coincides with the upper bound on $\|\Phi_q(x)\|$.

\subsubsection{Sampling for the Trotter remainder}

For a positive integer $m$, we factorize the truncated Trotter remainder as
\begin{align}
\label{eq:remainder-factorization}
    \tilde{\mc V}_{K}^{(q_0)}(x)
    =
    \biggl[
        \exp\biggl(
            \frac{1}{m}
            \sum_{q=K+1}^{q_0}\Phi_q(x)
        \biggr)
    \biggr]^m.
\end{align}
We call the factor repeated $m$ times on the right-hand side of
\eq{remainder-factorization} a fractional Trotter remainder. We construct an
LCQC for this map and apply $m$ independently sampled copies after each
product-formula step. The choice of $m$ is discussed in
\sec{complexity-analysis}.

Expanding the fractional Trotter remainder gives
\begin{align}
    &\exp\biggl(
        \frac{1}{m}
        \sum_{q=K+1}^{q_0}\Phi_q(x)
    \biggr)
     \nonumber\\
    ={}&\mc I+\frac{1}{m}\sum_{q=K+1}^{q_0}\Phi_q(x)
    +\sum_{j=2}^{\infty}\frac{1}{j!}
    \biggl(\frac{1}{m}\sum_{q=K+1}^{q_0}\Phi_q(x)\biggr)^j.
    \label{eq:one-factor-expansion}
\end{align}
The term with $j=1$ is linear in the BCH generator, whereas the terms with
$j\ge2$ are products of multiple BCH terms.
Define
\begin{align*}
    \lambda_{\mathrm{single}}
    =
    \frac{1}{m}
    \sum_{q=K+1}^{q_0}\lambda_q,
    \qquad
    \lambda_{\mathrm{multi}}
    =
    \sum_{j=2}^{\infty}
    \frac{\lambda_{\mathrm{single}}^j}{j!}.
\end{align*}
When $\lambda_{\mathrm{single}}$ is bounded by a constant,
\begin{align*}
    \lambda_{\mathrm{multi}}
    =
    e^{\lambda_{\mathrm{single}}}-1-\lambda_{\mathrm{single}}
    =
    \mathcal O(\lambda_{\mathrm{single}}^2).
\end{align*}
Applying the parameter-shift representation of each $\Phi_q(x)$ to
\eq{one-factor-expansion} gives an LCQC with $1$-norm
\begin{align*}
    \lambda_{\mathrm{unpaired}}
    =
    1+\lambda_{\mathrm{single}}+\lambda_{\mathrm{multi}} = 1+\mathcal O(\lambda_{\mathrm{single}})
\end{align*}
under the same condition.
The identity is sampled with probability $1/\lambda_{\mathrm{unpaired}}$;
otherwise, the sampler draws either a term from the linear part or a product
of multiple BCH terms according to its LCQC weight, and independently samples
the LCQC for each constituent BCH term. Every
nonidentity sampled channel is a composition of $\pi/4$ Pauli rotations and
therefore uses only Clifford gates. We refer to this construction as unpaired
HNCC.

We can improve the $1$-norm dependence on $\lambda_{\mathrm{single}}$ from
linear to quadratic by pairing the part of \eq{one-factor-expansion} that is
linear in the BCH generator with the identity channel $\mc I$. By the
preceding discussion, each $\Phi_q(x)$ is a linear
combination of $\ad_{\i P}$ for Pauli operators $P$. To illustrate the
reduction, consider a single term
$c\ad_{\i P}$, giving $\mc I+c\ad_{\i P}$. Expanding the unitary rotation
$\mc U_{e^{\theta\i P}}(\rho)=e^{\theta\i P}\rho e^{-\theta\i P}$ yields
\begin{align*}
    \mc U_{e^{\theta\i P}}
    =
    \cos^2\theta\,\mc I
    +
    \sin\theta\cos\theta\,\ad_{\i P}
    +
    \sin^2\theta\,\mc U_{\i P}.
\end{align*}
Setting $c=\tan\theta$ and rearranging terms gives a new LCQC identity. We
compare it with the parameter-shift rule:
\begin{align*}
    &\mc I+c\ad_{\i P}
    \\
    ={}&
    \begin{cases}
        \mc I+
        c\bigl(
            \mc U_{e^{\i\pi P/4}}
            -
            \mc U_{e^{-\i\pi P/4}}
        \bigr)
        &\text{parameter shift},\\
        (1+c^2)\mc U_{e^{\theta\i P}}
        -
        c^2\mc U_{\i P}
        &\text{identity pairing}.
    \end{cases}
\end{align*}
The $1$-norm of the parameter-shift representation is $1+2|c|$, whereas the
pairing identity yields $(1+c^2)+c^2=1+2c^2$. This reduces the $1$-norm from
$1+\mathcal O(c)$ to $1+\mathcal O(c^2)$.

For the full linear part, the sampler first draws a normalized Pauli-adjoint
term and applies the pairing identity with
$|c|=\lambda_{\mathrm{single}}/2$. Thus all sampled rotations in this part have
the same angle magnitude
$\theta=\tan^{-1}(\lambda_{\mathrm{single}}/2)$, independent of the sampled
Pauli operator. The resulting LCQC $1$-norm is
\begin{align*}
    1+\frac{\lambda_{\mathrm{single}}^2}{2}.
\end{align*}
Hence the LCQC for one fractional Trotter remainder has $1$-norm
\begin{align*}
    \lambda_{\mathrm{paired}}
    =
    1+\frac{\lambda_{\mathrm{single}}^2}{2}
    +\lambda_{\mathrm{multi}}
    =
    1+\mathcal O(\lambda_{\mathrm{single}}^2).
\end{align*}

\subsubsection{Gate complexity of a sampled channel}
\label{sec:classical-sample}

We first consider the part linear in the BCH terms. Each nonzero sampled
channel is a Pauli rotation by the common angle described above. Its generator
arises from a $q$-fold right-nested commutator of Pauli operators, each supported on at most
$k$ qubits. Its support therefore contains at most $kq\le kq_0$ qubits, so the
channel uses
$\mathcal O(kq_0)=\mathcal O(k\log(1/\eps))$ gates.

For a product $\prod_{r=1}^j\Phi_{q_r}(x)$ of multiple BCH terms, the sampled
channel is a composition of $\pi/4$ Pauli rotations and uses
$\mathcal O(k\sum_{r=1}^j q_r)=\mathcal O(ks)$ gates, where
$s=\sum_{r=1}^j q_r$ is the total BCH order. In particular, all rotations in
this part are Clifford gates. For $\ell\in[\nu]$ and $a\in[m]$, let
$s_{\ell,a}$ denote the total BCH order contributed by products of multiple
BCH terms in the $a$-th compensation sample after the $\ell$-th
product-formula step. We retain only outcomes satisfying
\begin{align}
\label{eq:total-order-condition}
    \sum_{\ell=1}^{\nu}\sum_{a=1}^{m}s_{\ell,a}\le s_0.
\end{align}
The proof of \thm{complexity} shows that $s_0$ can be chosen as
$\mathcal O(\log(1/\eps))$ while keeping the bias from discarded samples at
$\mathcal O(\eps\|O\|)$. Consequently, the retained products of multiple BCH
terms require $\mathcal O(k\log(1/\eps))$ gates over the entire circuit.
\subsection{Complexity analysis}
\label{sec:complexity-analysis}

We next choose the number of Trotter steps $\nu$ and bound the maximum gate
count per circuit. Throughout this subsection, $K$ is a fixed constant. When
$4q_0k\tilde g\le1/2$, \eq{lambda-q-summary} gives
\begin{align*}
    \lambda_{\mathrm{single}}
    \le
    \frac{N}{m}\sum_{q=K+1}^{q_0}(2q_0k\tilde g)^q
    =
    \mathcal O\Bigl(\frac{N}{m}(q_0k\tilde g)^{K+1}\Bigr).
\end{align*}
For $\lambda_{\mathrm{single}}=\mathcal O(1)$, the LCQC for one fractional
Trotter remainder has $1$-norm
$\lambda_{\mathrm{paired}}=1+\mathcal O(\lambda_{\mathrm{single}}^2)$. Since it is sampled $m$
times after each of the $\nu$ product-formula steps, the LCQC $1$-norm of the
full circuit is bounded by
\begin{align*}
    \lambda_{\mathrm{paired}}^{m\nu}
    &\le
    \exp\bigl(\mathcal O(m\nu\lambda_{\mathrm{single}}^2)\bigr)\\
    &\le
    \exp\Bigl(
        \mathcal O\Bigl(
            \frac{N^2(q_0kg_0t)^{2K+2}}
            {m\nu^{2K+1}}
        \Bigr)
    \Bigr),
\end{align*}
where we used $\tilde g=\mathcal O(g_0t/\nu)$. Requiring the exponent in this
bound to be $\mathcal O(1)$ preserves the $\mathcal O(\eps^{-2})$ repetition
count.
We take $q_0=\Theta(\log(1/\eps))$ to control the accumulated
BCH truncation error, as verified in the proof of \thm{complexity}.
Substituting this choice into the preceding $1$-norm bound, we choose
\[
    \nu
    =
    \Theta\Bigl(
        \Bigl(\frac{N^2}{m}\Bigr)^{\frac{1}{2K+1}}
        (kg_0t\log(1/\eps))^{1+\frac{1}{2K+1}}
    \Bigr),
\]
which keeps the $1$-norm $\lambda_{\mathrm{paired}}^{m\nu}$ bounded by a
constant. A sufficiently
large constant prefactor ensures $4q_0k\tilde g\le1/2$, as verified in the proof of
\thm{complexity}.

Each Trotter step uses $\mathcal O(\Gamma k)$ gates. After each Trotter step,
the $m$ compensation samples from the linear part use
$\mathcal O(mk\log(1/\eps))$ gates. With
$s_0=\mathcal O(\log(1/\eps))$ in \eq{total-order-condition}, the retained
products of multiple BCH terms contribute $\mathcal O(k\log(1/\eps))$ gates
over the entire circuit.
Since $m\nu\ge1$, this contribution is absorbed into the total cost of
sampling the linear part,
$\mathcal O(\nu m k\log(1/\eps))$. Hence every circuit uses
\begin{multline*}
    \mathcal O\Bigl(
        \Bigl(\frac{N^2}{m}\Bigr)^{\frac{1}{2K+1}}
        (kg_0t\log(1/\eps))^{1+\frac{1}{2K+1}}
        \\
        k\bigl(\Gamma+m\log(1/\eps)\bigr)
    \Bigr)
\end{multline*}
gates. Minimizing the dependence on $m$ over $1\le m\le N$ gives
$m=\Theta(\min\{N,\Gamma/\log(1/\eps)\})$, proving
\thm{fixed-complexity}.

\subsubsection{Unpaired HNCC}

Unpaired HNCC represents the linear part using the parameter-shift rule, so one
fractional Trotter remainder has LCQC $1$-norm
$\lambda_{\mathrm{unpaired}}=1+\mathcal O(\lambda_{\mathrm{single}})$.
Although $\lambda_{\mathrm{single}}$ scales as $1/m$, the fractional
remainder is sampled $m$ times after each product-formula step. Increasing $m$
therefore does not reduce the leading contribution to the full-step LCQC
$1$-norm, and we take $m=1$. The preceding bound then gives
$\lambda_{\mathrm{unpaired}}=1+\mathcal O(N(q_0kg_0x)^{K+1})$. Keeping
$\lambda_{\mathrm{unpaired}}^{2\nu}=\mathcal O(1)$ therefore requires
$\nu N(q_0kg_0x)^{K+1}=\mathcal O(1)$, so we take
\[
    \nu
    =
    \Theta\Bigl(
        N^{1/K}
        (kg_0t\log(1/\eps))^{1+1/K}
    \Bigr).
\]
Thus unpaired HNCC has the same $N$- and $t$-dependence as the underlying
$K$-th order product formula while achieving polylogarithmic precision
dependence.

To obtain a maximum gate-count guarantee, we retain only samples for which
the total BCH order of all sampled compensation terms is
$\mathcal O(\log(1/\eps))$. This limits every circuit to
$\mathcal O(k\log(1/\eps))$ compensation gates and introduces only
$\mathcal O(\eps)$ additional bias. Since the unpaired representation uses
only $\pi/4$ Pauli rotations, these compensation gates are Clifford gates and
require no $T$ gates. The formal statement and proof are given in
\thm{unpaired-global-order}.

\section{Notation}
\label{sec:prelim}
For a positive integer $N$, we write $[N]=\{1,\ldots,N\}$. For an operator
$A$, let $\supp(A)$ denote the set of sites on which $A$ acts nontrivially.
We use $\|\cdot\|$ for the operator norm of matrices and the diamond norm of
superoperators.

We set $[X_1]:=X_1$ and, for $q\ge2$, write
$[X_1,\ldots,X_q]:=[X_1,[X_2,\ldots,[X_{q-1},X_q]\ldots]]$ for the
right-nested commutator. We also define $\ad_A(B)=[A,B]$.

We let $\mathcal P_N=\{I,X,Y,Z\}^{\otimes N}$ denote the set of unsigned
$N$-qubit Pauli operators, with signs absorbed into their coefficients. We
write $\pm\i\mathcal P_N=\{\pm\i P:P\in\mathcal P_N\}$, and let $\CPTP_N$
denote the set of $N$-qubit quantum channels. For any unitary $V$, we denote
the corresponding channel by $\mc U_V=V(\cdot)V^\dagger$.

\section{HNCC construction}
\label{sec:hncc-construction}
We now give the detailed construction of HNCC.
For a set of operators $\{h_\gamma\}_\gamma$ on a finite site set $\Lambda$, we say
that the set is $g$-extensive if
\begin{align*}
    \sum_{\gamma:\supp(h_\gamma)\ni j}\|h_\gamma\|
    \le g,
    \qquad \forall j\in\Lambda.
\end{align*}
A $k$-local Hamiltonian is $g_0$-extensive if it admits a decomposition into
a $g_0$-extensive set of terms, each supported on at most $k$ sites.
Let
\begin{align}
\label{eq:def-H}
    H = \sum_{\gamma=1}^{\Gamma} \alpha_{\gamma} P_{\gamma}
\end{align}
be the Pauli decomposition of the $k$-local Hamiltonian, where each
$P_\gamma\in\mathcal{P}_N$ has support size at most $k$. Identity Pauli terms
are omitted, since they contribute only a global phase.
We assume that this decomposition satisfies the $g_0$-extensiveness condition
\begin{align}
\label{eq:def-extensive}
\sum_{\gamma: \mathrm{supp}(P_{\gamma}) \ni j} |\alpha_{\gamma}|
\le g_0, \quad \forall j \in [N].
\end{align}
We partition the Hamiltonian into $L$ groups of commuting terms:
\begin{align}
\label{eq:local-decompose-H-l}
H = \sum_{\ell=1}^{L} H_{\ell}, \quad H_{\ell} = \sum_{\gamma} \alpha_{\ell, \gamma} P_{\ell, \gamma},
\end{align}
such that $[P_{\ell, \gamma}, P_{\ell, \gamma'}] = 0$ for all $\gamma, \gamma'$ within each group $\ell$. The $g_0$-extensive condition \eq{def-extensive} then reads
\begin{align}
\label{eq:extensive-decompose}
\sum_{\substack{\ell\in[L],\,\gamma:\\ \supp(P_{\ell,\gamma})\ni j}}
\|\alpha_{\ell,\gamma}P_{\ell,\gamma}\|
=\sum_{\substack{\ell\in[L],\,\gamma:\\ \supp(P_{\ell,\gamma})\ni j}}
|\alpha_{\ell,\gamma}| \le g_0,
\end{align}
for any $j\in [N]$.

We consider the evolution channel
\begin{align*}
    \mc U(t)=e^{-\i t\ad_H}.
\end{align*}
For a step size $x$, let $\mc S_K(x)$ denote the $K$-th order product-formula
channel
\begin{align*}
    \mc S_K(x)
    =
    \prod_{v=1}^{\kappa_KL} e^{-\i a_v x\ad_{H_{\ell_v}}},
\end{align*}
where $\kappa_K$ is the number of product-formula stages per group,
$a_{\max}=\max_v |a_v|$, and $\ell_v\in[L]$ specifies the ordering of the
product. For the standard Suzuki--Trotter formula,
$\kappa_K=2\cdot 5^{K/2-1}$ and $a_{\max}\le K/3^{K/2}$~\cite{wiebe2010higher}. We define the
multiplicative Trotter remainder by
\begin{align*}
    \mc V_K(x)=\mc U(x)\mc S_K^\dagger(x).
\end{align*}
Our goal is to compensate this remainder after each Trotter step, so that
$\mc V_K(x)\mc S_K(x)=\mc U(x)$ in the ideal case.

We first state the randomized LCQC sampling rule and then construct an LCQC
for the Trotter remainder from its truncated BCH expansion. We next give the
one-step sampler and the full estimator.


\subsection{Randomized linear combination of quantum channels}

Suppose that a superoperator $\mc M$ admits a linear combination of quantum
channels (LCQC) $\mc M=\sum_j\beta_j\mc N_j$ with $1$-norm
$\lambda=\sum_j|\beta_j|$.
Sampling an index $j$ with probability $p(j)=|\beta_j|/\lambda$ and setting
$(\eta,\mc N)=(\sgn(\beta_j),\mc N_j)$ gives
\begin{align*}
    \E[\eta \mc N]
    =
    \sum_j p(j)\sgn(\beta_j)\mc N_j
    =
    \sum_j \frac{\beta_j}{\lambda}\mc N_j
    =
    \frac{1}{\lambda}\mc M.
\end{align*}
The following lemma bounds the number of samples required to estimate
$\tr[O\mc M(\rho_0)]$ for any observable $O$ and initial state $\rho_0$,
provided such a sampling procedure exists.

\begin{lemma}
\label{lem:random-channel}
    Given an $N$-qubit superoperator $\mc M$ and a normalization factor $\lambda > 0$, assume that we can sample a tuple $(\eta,\mc N)$ with $\eta\in[-1,1]$ and
    $\mc N$ either an $N$-qubit quantum channel or the zero map, such that
    \begin{align*}
        \E[\eta \mc N] = \frac{1}{\lambda}\mc M.
    \end{align*}
    Given $\eps > 0$, an observable $O$, and an initial state $\rho_0$, there exists an algorithm estimating $\tr[O \mc M(\rho_0)]$ to precision $\eps\|O\|$ with success probability at least $1-\delta$ using 
    \begin{align*}
        \mathcal{O}\Bigl(\frac{\lambda^2}{\eps^2}\log(1/\delta)\Bigr)
    \end{align*}
    independent samples, with a quantum circuit executed only when the sampled
    map $\mc N$ is nonzero.
\end{lemma}

\begin{proof}
    In each run, we sample a tuple $(\eta, \mc N)$. If $\mc N$ is a quantum channel, we apply it to the initial state $\rho_0$, measure the observable $O$, and let $v$ be the measurement outcome. If $\mc N = 0$, we do not execute any quantum circuit and set $v = 0$. In both cases, the expectation of $v$ conditional on the sampled $\mc N$ is $\tr[O \mc N(\rho_0)]$, and the value is bounded by $|v| \le \|O\|$. 
    
    Set $X=\lambda\eta v$. Its expectation is
    \begin{align*}
        \E[X] &= \E\big[ \lambda \eta \tr[O \mc N(\rho_0)] \big] \\
        &= \lambda \tr\big[ O \E[\eta \mc N](\rho_0) \big] = \tr[O \mc M(\rho_0)].
    \end{align*}
    Since $|v| \le \|O\|$ and $|\eta| \le 1$, the random variable $X$ always takes values in the interval $[-\lambda\|O\|, \lambda\|O\|]$. 
    
    By performing $n = 2\lambda^2 \eps^{-2} \ln(2/\delta)$ independent sampling runs, we obtain the empirical average $u = \frac{1}{n}\sum_{i=1}^n X_i$. Applying Hoeffding's inequality yields 
    \begin{align*}
        \Pr\big[|u - \tr[O\mc M(\rho_0)]| \ge \eps \|O\|\big] &\le 2\exp\Bigl(-\frac{2n(\eps\|O\|)^2}{(2\lambda\|O\|)^2}\Bigr) \\
        &= 2\exp\Bigl(-\frac{n\eps^2}{2\lambda^2}\Bigr) \le \delta.
    \end{align*}
    This completes the proof.
\end{proof}

\subsection{BCH expansion and truncation of the Trotter remainder}
In this subsection, we express the Trotter remainder $\mc V_K(x)$ approximately as a linear combination of nested commutators of $P_{\gamma}$.

We first recall the explicit BCH formula used below. For $\sigma\in S_q$,
define the reordering operator $\mc R_\sigma\{\cdot\}$ by
\begin{align*}
    \mc R_\sigma\big\{[O_q,\ldots,O_2,O_1]\big\}
    =
    [O_{\sigma(q)},\ldots,O_{\sigma(2)},O_{\sigma(1)}].
\end{align*}
Formally, the BCH expansion of $e^{X_M}\cdots e^{X_2}e^{X_1}$ can be written as
\begin{align}
\label{eq:def-BCH}
    e^{X_M}\cdots e^{X_2}e^{X_1}
    =
    \exp\biggl(\sum_{q=1}^{\infty}\Phi_q\biggr),
\end{align}
where $\Phi_q$ is homogeneous of degree $q$ in the generators and is a
weighted sum of $q$-fold right-nested commutators of the
$X_j$~\cite{arnal2021note}:
\begin{align}
\label{eq:BCH}
    \Phi_q
    =
    \begin{multlined}[t]
    \sum_{\substack{p_v\ge 0\\p_1+\cdots+p_M=q}}
    \frac{1}{p_1!\cdots p_M!}
    \sum_{\sigma\in S_q}c_{\sigma,q}\\
    {}
    \mc R_\sigma\bigl\{
    [\underbrace{X_M,\ldots,X_M}_{p_M},\ldots,
     \underbrace{X_1,\ldots,X_1}_{p_1}]
    \bigr\},
    \end{multlined}
\end{align}
with
\begin{align}
\label{eq:def-c-sigma}
    c_{\sigma,q}
    =
    \frac{1}{q^2}
    \frac{(-1)^{d_\sigma}}{\binom{q-1}{d_\sigma}},
\end{align}
where $d_\sigma$ is the number of $i\in[q-1]$ satisfying
$\sigma(i)>\sigma(i+1)$. The equality in \eq{def-BCH} is understood as a
formal power-series identity unless convergence is stated.

Set $M:=\kappa_KL+1$. In the descending-product notation introduced above,
the Trotter remainder is
\begin{align*}
    \mc V_K(x)
    =
    e^{-\i x\ad_H}
    \prod_{v=M-1}^{1}e^{\i a_vx\ad_{H_{\ell_v}}}.
\end{align*}
To apply the BCH formula \eq{def-BCH}, we relabel these $M$ factors in the
order in which their channels are applied and write
\begin{align}
    \mc V_K(x)=\prod_{v=1}^{M}e^{-\i\ad_{\tilde H_v}},
\end{align}
where
\begin{align}
\label{eq:def-tilde-H}
    \tilde H_M
    &={}
    xH,
    \nonumber\\
    \tilde H_{M-v}
    &={}
    -a_vxH_{\ell_v},
    \qquad v\in[M-1].
\end{align}
Let $\Phi_q(x)$ denote the corresponding BCH term of order $q$, which is
homogeneous of degree $q$ in $x$. Since $\mc S_K$ is a $K$-th
order product formula, $\Phi_1(x)=\cdots=\Phi_K(x)=0$. We define the formal
BCH generator by
\begin{align}
\label{eq:def-VK-BCH}
    \Phi(x)
    &:={}
    \sum_{q=K+1}^{\infty}\Phi_q(x),
    \\
    \mc V_K(x)
    &={}
    \prod_{v=1}^{M}e^{-\i\ad_{\tilde H_v}}
    =
    \exp\bigl(\Phi(x)\bigr).
\end{align}
The Pauli decompositions of $H_{\ell}$ and $H$ in \eq{local-decompose-H-l} and \eq{def-H} imply that each $\tilde{H}_v$ admits a local Pauli decomposition denoted by \begin{align}
\label{eq:decompose-tilde-H}
    \tilde{H}_v = \sum_{\gamma} \tilde{\alpha}_{v,\gamma} \tilde{P}_{v,\gamma}=:\sum_{\gamma}\tilde{h}_{v,\gamma}, \quad \forall v\in[M].
\end{align} 
The dependence of $\tilde H_v$, $\tilde\alpha_{v,\gamma}$, and
$\tilde h_{v,\gamma}$ on the step size $x$ is kept implicit unless needed.
We take $\tilde P_{v,\gamma}\in\mathcal P_N$ and allow
$\tilde{\alpha}_{v,\gamma}\in\mathbb R$, with the sign carried by the
coefficient.

\begin{lemma}
\label{lem:bound-Phi-q}
    Let $\tilde g := (a_{\max}\kappa_K+1)g_0x$. The $q$-th order term in the BCH expansion of $\mc V_K(x)$ in \eq{def-VK-BCH} is bounded by 
    \begin{align}
    \label{eq:lambda-q}
        \|\Phi_q(x)\|\le \frac{q!(2k\tilde g)^qk^{-1} N}{q^2}=:\lambda_q.
    \end{align}
\end{lemma}

\begin{proof}
    We first bound the total interaction strength per site for the collection of local terms $\{\tilde{h}_{v,\gamma}\}_{v,\gamma}$ defined in \eq{decompose-tilde-H}. For any site $j\in [N]$, we have
    \begin{align}
    &\sum_{v=1}^{M} \sum_{\gamma: \mathrm{supp}(\tilde{h}_{v,\gamma}) \ni j}
    \| \tilde{\alpha}_{v,\gamma} \tilde{P}_{v,\gamma}\| \nonumber\\
    \le{}& x\sum_{v=1}^{\kappa_K L} |a_v|
    \sum_{\gamma:\mathrm{supp}(P_{\ell_v,\gamma})\ni j} |\alpha_{\ell_v,\gamma}|
    + x\sum_{\gamma:\mathrm{supp}(P_{\gamma})\ni j} |\alpha_{\gamma}| \nonumber\\
    \le{}& a_{\max}x\sum_{v=1}^{\kappa_K L}
    \sum_{\gamma:\mathrm{supp}(P_{\ell_v,\gamma})\ni j} |\alpha_{\ell_v,\gamma}|
    + g_0x \nonumber\\
    ={}& a_{\max}x\kappa_K\sum_{\ell=1}^L
    \sum_{\gamma:\mathrm{supp}(P_{\ell,\gamma})\ni j} |\alpha_{\ell,\gamma}|
    + g_0x \nonumber\\
    \le{}& (a_{\max}\kappa_K+1)g_0x = \tilde g .
    \label{eq:extensive-tilde-H}
    \end{align}
    The second inequality uses $|a_v|\le a_{\max}$ and
    \eq{def-extensive}. The equality holds because each index
    $\ell\in[L]$ appears exactly $\kappa_K$ times in
    $\{\ell_v\}_{v=1}^{\kappa_KL}$, and the final inequality follows from
    \eq{extensive-decompose}. Hence \eq{extensive-tilde-H} shows that the
    collection of local terms $\{\tilde h_{v,\gamma}\}_{v,\gamma}$ is
    $\tilde g$-extensive.

    By \cite[Proposition~5]{aftab2024multi}, if
    $\exp(\sum_{q\ge1}\Phi_q)=\prod_{v=1}^M e^{X_v}$, then
    \begin{align*}
        \|\Phi_q\|
        \le
        \frac{1}{q^2}
        \sum_{v_1,\ldots,v_q=1}^M
        \|[X_{v_q},\ldots,X_{v_1}]\|.
    \end{align*}
    Applying this bound to the local adjoint generators gives
    \begin{align*}
        \|\Phi_q(x)\| &\le \frac{1}{q^2} \sum_{(v_1,\gamma_1), \ldots,(v_q,\gamma_q)}\big\|[\ad_{\tilde{h}_{v_q, \gamma_q}}, \ldots,\ad_{\tilde{h}_{v_1, \gamma_1}}]\big\| \nonumber\\
        &= \frac{1}{q^2}\sum_{(v_1,\gamma_1), \ldots,(v_q,\gamma_q)}\big\|\ad_{[\tilde{h}_{v_q, \gamma_q}, \ldots, \tilde{h}_{v_1, \gamma_1}]}\big\|\\
        &\le \frac{2}{q^2}\sum_{(v_1,\gamma_1), \ldots,(v_q,\gamma_q)} \big\|[\tilde{h}_{v_q, \gamma_q}, \ldots, \tilde{h}_{v_1, \gamma_1}]\big\|.
    \end{align*}
    Since each $\tilde{h}_{v,\gamma}$ is supported on at most $k$ sites and the collection is $\tilde g$-extensive, by \cite[Corollary 2]{wang2026lindbladian}, the sum of the nested commutators is bounded by 
    \begin{align*}
        \|\Phi_q(x)\| &\le \frac{2}{q^2}\frac{q!(2k\tilde g)^qN}{2k}
        \le \frac{q!(2k\tilde g)^q k^{-1} N}{q^2} = \lambda_q.
    \end{align*}
    This completes the proof.
\end{proof}
The factorial growth in this upper bound prevents it from controlling the
full series $\sum_{q=K+1}^{\infty}\Phi_q(x)$ for any fixed $x>0$. We
therefore consider the exponential of the $q_0$-th order truncated BCH
expansion \begin{align*}
    \tilde{\mc V}_{K}^{(q_0)}(x) = \exp\biggl(\sum_{q= K+1}^{q_0}\Phi_q(x)\biggr).
\end{align*}
The following lemma is the Hamiltonian specialization of
\cite[Theorem~2]{wang2026lindbladian}.
\begin{lemma}
 \label{lem:general-BCH-truncation}
Given a sequence of Hamiltonians $H_1, \ldots, H_M$, let
$\{\Phi_q\}_{q\ge1}$ denote the BCH terms in the expansion of
$\prod_{v=1}^M e^{-\i\ad_{H_v}}$. Define the doubly right-nested
commutator bound \begin{align*}
\alpha_{\mathrm{comm}}^{(q_1,\ldots,q_d)}
={}&
\begin{multlined}[t]
\sum_{v_{1,1},\ldots,v_{d,q_d}=1}^{M}
\Big\|
\Big[
    [\ad_{H_{v_{1,1}}},\ldots,\ad_{H_{v_{1,q_1}}}],\ldots,\\
    [\ad_{H_{v_{d,1}}},\ldots,\ad_{H_{v_{d,q_d}}}]
\Big]
\Big\|.
\end{multlined}
\end{align*} For any positive integer $q_0$, the truncation error of the $q_0$-th order BCH expansion is bounded by 
 \begin{align}
\label{eq:def-beta-comm-sum}
    &\biggl\|\exp\biggl(\sum_{q=1}^{q_0} \Phi_q\biggr) - \prod_{v = 1}^Me^{-\i \ad_{H_{v}}} \biggr\| \nonumber\\
    \le{}&  \sum_{d=1}^{\infty}\frac{1}{d!} \sum_{\substack{1\le q_1, \ldots, q_{d}\le q_0\\q_1+\cdots+q_{d} \ge q_0+1}}\alpha_{\mathrm{comm}}^{(q_1, \ldots, q_{d})}.
\end{align}
\end{lemma}
\begin{proof}
    Apply \cite[Theorem~2]{wang2026lindbladian} to the generators
    $\mc L_v=-\i\ad_{H_v}$. That theorem bounds the error by the
    commutator tail in \eq{def-beta-comm-sum} multiplied by the norm of the
    truncated evolution. In the present Hamiltonian setting, this norm is one.

    Indeed, Hamiltonian Liouvillians are closed under commutators: for Hermitian
    $A$ and $B$,
    \[
        [-\i\ad_A,-\i\ad_B]
        =
        -\i\ad_{-\i[A,B]},
    \]
    and $-\i[A,B]$ is Hermitian. Hence the truncated logarithm in the
    application of that theorem has the form
    $-\i\ad_{G_{(q_0)}(\tau)}$ for some Hermitian
    $G_{(q_0)}(\tau)$. Its exponential and the exact evolution are unitary
    channels, so the additional norm factor equals $1$.
\end{proof}
The following lemma provides an estimate of the doubly right-nested commutator bound for local operators. 
\begin{lemma}[{\cite[Theorem~9]{wang2026lindbladian}}]
\label{lem:doubly-nested-bound}
    Let $h_1, \ldots, h_{M}$ be operators on \(N\) sites, each supported on at most $k$ sites. Suppose that $\{h_v\}_{v=1}^M$ is $\tilde g$-extensive. For a sequence of positive integers $q_1, \ldots, q_{d}$, define the cumulative counts $P_{r} = \sum_{j=r}^{d} q_j$ for $r \in [d]$ and $P_{d+1}=1$. Then, the following bound holds:
    \begin{equation}
    \label{eq:double-comm-k-local-N}
    \begin{multlined}
    \sum_{v_{1,1},\ldots,v_{d,q_d}=1}^M
    \big\|
    [
        [h_{v_{1,1}},\ldots,h_{v_{1,q_1}}],\ldots,\\
        [h_{v_{d,1}},\ldots,h_{v_{d,q_d}}]
    ]
    \big\|
    \le
    \frac{1}{2kq_d}
    \biggl(\prod_{r=1}^{d} P_{r+1} q_r! (2k\tilde g)^{q_r}\biggr)N .
    \end{multlined}
    \end{equation}
\end{lemma}
Combining \lem{general-BCH-truncation} and \lem{doubly-nested-bound} gives the following upper bound for the BCH truncation error. 
\begin{lemma}
\label{lem:truncation-BCH}
    Let $H = \sum_{\ell=1}^L H_\ell$ be the $k$-local and $g_0$-extensive Hamiltonian defined in \eq{def-H} and partitioned as in \eq{local-decompose-H-l}. Let $\tilde g = (a_{\max}\kappa_K+1)g_0x$. The BCH expansion of $\mc V_K(x)$ defined in \eq{def-VK-BCH}, truncated after order $q_0$, satisfies
    \begin{align}
    \label{eq:truncate-error}
    \biggl\|\mc V_K(x)-\exp\biggl(\sum_{q= K+1}^{q_0}\Phi_q(x)\biggr)\biggr\| \le 2(4eq_0k\tilde g)^{q_0+1}N,
    \end{align}
    provided that $8eq_0k\tilde g \le 1$.
\end{lemma} 
\begin{proof}
    Since $\Phi_1(x)=\cdots=\Phi_K(x)=0$, we may write the truncated
    exponent as $\sum_{q=1}^{q_0}\Phi_q(x)$. By
    \lem{general-BCH-truncation},
    \begin{align}
    \label{eq:BCH-error-local}
        \big\|\tilde{\mc V}_K^{(q_0)}(x) - \mc V_K(x)\big\|&=\biggl\|\exp\biggl(\sum_{q=1}^{q_0} \Phi_q(x)\biggr) - \mc V_K(x)\biggr\|\nonumber\\
        &\le  \sum_{d=1}^{\infty}\frac{1}{d!} \sum_{\substack{1\le q_1, \ldots, q_{d}\le q_0\\q_1+\cdots+q_{d} \ge q_0+1}}\alpha_{\mathrm{comm}}^{(q_1, \ldots, q_{d})}(x),
    \end{align}
    where $\alpha_{\mathrm{comm}}^{(q_1, \ldots, q_{d})}(x)$ is the doubly right-nested commutator bound for the Hamiltonians $\{\tilde{H}_v\}_v$. Let $q=\sum_{j=1}^{d}q_j$. Then
    \begin{widetext}
    \begin{align*}
    \alpha_{\mathrm{comm}}^{(q_1,\ldots,q_d)}(x)
    ={}&
    \sum_{v_{1,1},\ldots,v_{d,q_d}=1}^{M}
    \Big\|
    \Big[
        [\ad_{\tilde H_{v_{1,1}}},\ldots,\ad_{\tilde H_{v_{1,q_1}}}],\ldots,
        [\ad_{\tilde H_{v_{d,1}}},\ldots,\ad_{\tilde H_{v_{d,q_d}}}]
    \Big]
    \Big\| \\
    \le{}&
    2\sum_{v_{1,1},\ldots,v_{d,q_d}=1}^{M}
    \Big\|
    \Big[
        [\tilde H_{v_{1,1}},\ldots,\tilde H_{v_{1,q_1}}],\ldots,
        [\tilde H_{v_{d,1}},\ldots,\tilde H_{v_{d,q_d}}]
    \Big]
    \Big\| \\
    \le{}&
    2\sum_{v_{1,1},\ldots,v_{d,q_d}=1}^{M}
    \sum_{\gamma_{1,1},\ldots,\gamma_{d,q_d}}
    \Big\|
    \Big[
        [\tilde h_{v_{1,1},\gamma_{1,1}},\ldots,
         \tilde h_{v_{1,q_1},\gamma_{1,q_1}}],\ldots,
        [\tilde h_{v_{d,1},\gamma_{d,1}},\ldots,
         \tilde h_{v_{d,q_d},\gamma_{d,q_d}}]
    \Big]
    \Big\| \\
    \le{}& q^d\biggl(\prod_{r=1}^{d}q_r!\biggr)(2k\tilde g)^qN .
    \end{align*}
    \end{widetext}
    The first inequality uses $\|\ad_A\|\le 2\|A\|$, and the second
    uses the triangle inequality over the Pauli decompositions. For the last
    inequality, \lem{doubly-nested-bound} applies to the $\tilde g$-extensive
    collection $\{\tilde h_{v,\gamma}\}_{v,\gamma}$. Since
    $P_{r+1}\le q$ and $kq_d\ge1$, its prefactor is at most $q^d$.
Substituting this bound for
$\alpha_{\mathrm{comm}}^{(q_1,\ldots,q_d)}(x)$ into
\eq{BCH-error-local} and grouping the terms by
$q=\sum_{r=1}^d q_r$ yields
\begin{align*}
    &\big\|\tilde{\mc V}_K^{(q_0)}(x)-\mc V_K(x)\big\|\\
    \le{}&
    \sum_{q=q_0+1}^{\infty}
    \sum_{d=1}^{q}\frac{q^d}{d!}
    \sum_{\substack{
        1\le q_1,\ldots,q_d\le q_0\\
        q_1+\cdots+q_d=q}}
    \biggl(\prod_{r=1}^d q_r!\biggr)(2k\tilde g)^qN\\
    \le{}&
    \sum_{q=q_0+1}^{\infty}
    \sum_{d=1}^{q}\frac{q^d}{d!}
    \binom{q-1}{d-1}
    q_0^q(2k\tilde g)^qN\\
    \le{}&
    \sum_{q=q_0+1}^{\infty}
    e^q(4q_0k\tilde g)^qN\\
    \le{}&
    (4eq_0k\tilde g)^{q_0+1}N
    \sum_{j=0}^{\infty}2^{-j}\\
    \le{}&
    2(4eq_0k\tilde g)^{q_0+1}N.
\end{align*}
Here we used
$\prod_{r=1}^d q_r!\le q_0^q$,
$\sum_{d=1}^{q}q^d/d!\le e^q$, and
\begin{align*}
    \sum_{\substack{
        1\le q_1,\ldots,q_d\le q_0\\
        q_1+\cdots+q_d=q}}1
    \le
    \binom{q-1}{d-1}
    \le
    2^{q-1}.
\end{align*}
\end{proof}

\subsection{One-step compensation and total simulation}

After truncating the BCH expansion, we write
\begin{align*}
    \tilde{\mc V}_{K}^{(q_0)}(x)
    =
    \biggl[
    \exp\biggl(
        \frac{1}{m}\sum_{q=K+1}^{q_0}\Phi_q(x)
    \biggr)
    \biggr]^m
\end{align*}
for a positive integer $m$. Define
\begin{equation}
\label{eq:lambda-single}
\begin{aligned}
    \lambda_{\mathrm{single}}
    &={}
    \frac{1}{m}\sum_{q=K+1}^{q_0}\lambda_q,\quad
    \lambda_{\mathrm{multi}}
    ={}
    \sum_{j=2}^{\infty}
    \frac{\lambda_{\mathrm{single}}^j}{j!},\\
    \lambda_{\mathrm{paired}}
    &={}
    1+\frac{\lambda_{\mathrm{single}}^2}{2}
    +\lambda_{\mathrm{multi}},\\
    \lambda_{\mathrm{unpaired}}
    &={}
    1+\lambda_{\mathrm{single}}+\lambda_{\mathrm{multi}}
    =
    \sum_{j=0}^{\infty}\frac{\lambda_{\mathrm{single}}^j}{j!}
    =
    e^{\lambda_{\mathrm{single}}}.
\end{aligned}
\end{equation}
The next two theorems give paired and unpaired LCQC samplers for the
fractional Trotter remainder
\begin{align}
\label{eq:fractional-remainder}
    \exp\biggl(
        \frac{1}{m}\sum_{q=K+1}^{q_0}\Phi_q(x)
    \biggr).
\end{align}
Their LCQC $1$-norms are $\lambda_{\mathrm{paired}}$ and
$\lambda_{\mathrm{unpaired}}$, respectively. We apply
the chosen sampler independently $m$ times after every product-formula step.
Throughout the theorem statements below, $K$ is fixed, and implicit constants
may depend on $K$.
\begin{theorem}
\label{thm:hncc}
    Let $H$ be a $k$-local and $g_0$-extensive Hamiltonian on $N$ qubits,
    let $q_0\ge K+1$, and let $m\ge1$. Let $\mc S_K$ be the
    $K$-th order product formula defined above.
    \algo{HNCC} samples $\eta\in\{-1,1\}$,
    $\mc N\in\CPTP_N\cup\{0\}$, and a nonnegative integer $s$ such that
    \begin{align*}
        \E[\eta\mc N]
        =
        \frac{1}{\lambda_{\mathrm{paired}}}
        \exp\biggl(
            \frac{1}{m}\sum_{q=K+1}^{q_0}\Phi_q(x)
        \biggr).
    \end{align*}
    Every nonzero output $\mc N$ can be implemented using
    $\mathcal O(kq_0)$ gates if $s=0$ and $\mathcal O(ks)$ gates if $s>0$.
    If $\lambda_{\mathrm{single}}$ is bounded by a constant, one call to
    \algo{HNCC} has expected
    classical cost $\poly(N,\Gamma,q_0)$.
\end{theorem}

\begin{theorem}[Unpaired one-step compensation]
\label{thm:unpaired-hncc}
    Under the assumptions of \thm{hncc}, unpaired HNCC samples
    $\eta\in\{-1,1\}$, $\mc N\in\CPTP_N\cup\{0\}$, and a nonnegative
    integer $s$ such that
    \begin{align*}
        \E[\eta\mc N]
        =
        \frac{1}{\lambda_{\mathrm{unpaired}}}
        \exp\biggl(
            \frac{1}{m}\sum_{q=K+1}^{q_0}\Phi_q(x)
        \biggr).
    \end{align*}
    The identity output has $s=0$ and requires no gates. Every other nonzero
    output uses $\mathcal O(ks)$ Clifford gates. If
    $\lambda_{\mathrm{single}}$ is bounded by a constant, one sample has
    expected classical cost
    $\poly(N,\Gamma,q_0)$.
\end{theorem}

The proofs of \thm{hncc} and \thm{unpaired-hncc}, together with the explicit
sampling routines for paired HNCC, are given in \sec{implementation}.
Algorithm~\ref{algo:compensated-simulation} gives the full HNCC estimator, and
\fig{compensated-flow} shows the circuit used in one repetition.

\begin{algorithm}[H]
\caption{HNCC Simulation Estimator}
\label{algo:compensated-simulation}
\KwInput{Product-formula data $\mathsf{PFData}$, initial state $\rho_0$,
observable $O$, simulation time $t$, and target precision $\eps$.}
\KwOutput{An estimate $\widehat{o}$ of
$\tr[O\mc U(t)(\rho_0)]$.}

Set $m$, $q_0$, $\nu$, and $s_0$ as in the proof of
\thm{complexity}\;
Set $x\gets t/\nu$\;
Set
$\lambda_{\mathrm{paired}}$ according to \eq{lambda-single}\;
Set $R\gets\Theta(\lambda_{\mathrm{paired}}^{2m\nu}\eps^{-2})$ and
$\widehat{o}\gets0$\;

\For{$r=1$ \KwTo $R$}{
    Independently sample
    $(\eta_{\ell,a},\mc N_{\ell,a},s_{\ell,a})
    \gets
    \textsf{HNCC}(\mathsf{PFData},x,q_0,m)$
    for all $\ell\in[\nu]$ and $a\in[m]$\;

    \eIf{some $\mc N_{\ell,a}=0$ or
    $\sum_{\ell=1}^{\nu}\sum_{a=1}^{m}s_{\ell,a}>s_0$
    \label{lin:total-order-cutoff}}{
        $X_r\gets0$\;
    }{
        Prepare the initial state $\rho_0$\;
        \For{$\ell=1$ \KwTo $\nu$}{
            Apply $\mc S_K(x)$\;
            \For{$a=1$ \KwTo $m$}{
                Apply $\mc N_{\ell,a}$\;
            }
        }
        Measure $O$ and let $v_r$ be the outcome\;
        $X_r\gets
        \lambda_{\mathrm{paired}}^{m\nu}
        \bigl(
            \prod_{\ell=1}^{\nu}
            \prod_{a=1}^{m}
            \eta_{\ell,a}
        \bigr)v_r$\;
    }
    $\widehat{o}\gets\widehat{o}+X_r/R$\;
}
\Return $\widehat{o}$\;
\end{algorithm}

\begin{figure*}[t]
    \centering
    \includegraphics[page=1,width=0.75\textwidth]{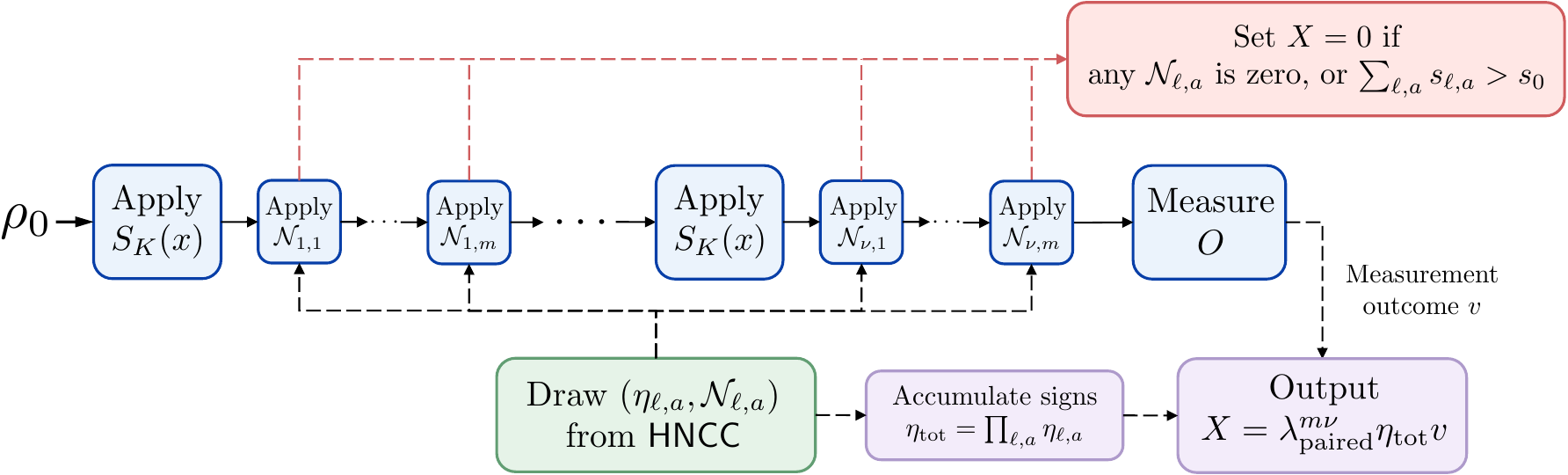}
    \caption{Circuit used in one repetition of
    \algo{compensated-simulation}. Each product-formula step is followed by
    $m$ independently sampled HNCC channels. Their signs are accumulated in
    $\eta_{\rm tot}$, and the measurement outcome is rescaled by
    $\lambda_{\mathrm{paired}}^{m\nu}\eta_{\rm tot}$. If any sampled channel is zero or the
    total BCH order exceeds $s_0$, the algorithm outputs $X=0$ without running
    the circuit.}
    \label{fig:compensated-flow}
\end{figure*}

We next analyze the complexity of the full simulation algorithm. The
fractional-remainder factorization is valid for every positive integer $m$.
To keep the BCH truncation bound independent of $m$, we restrict the complexity
analysis to $1\le m\le N$. The gate-count bound derived below is minimized, up
to a constant factor, by
\begin{align}
\label{eq:choice-m}
    m
    =
    \Theta\Bigl(
        \min\Bigl\{
            N,\,
            \frac{\Gamma}{\log(1/\eps)}
        \Bigr\}
    \Bigr),
\end{align}
rounded to an integer in $[1,N]$.

\begin{theorem}[High-order nested-commutator compensation]
\label{thm:complexity}
    Let $H=\sum_{\gamma=1}^{\Gamma}\alpha_\gamma P_\gamma$ be a
    $k$-local Hamiltonian on $N$ qubits satisfying the
    $g_0$-extensiveness condition. For any fixed $K$, let $\mc S_K$ be a
    $K$-th order product formula. Let $t>0$ and $0<\eps<1/2$, and suppose
    that $\Gamma=\Omega(\log(1/\eps))$ and $kg_0t\ge1$.
    Given any initial state $\rho_0$ and observable $O$,
    \algo{compensated-simulation} estimates
    $\tr[O\mc U(t)(\rho_0)]$ to additive precision $\eps\|O\|$ using
    $\mathcal O(\eps^{-2})$ independent circuit repetitions. For $m$ chosen
    as in \eq{choice-m}, the maximum gate count per circuit is
    \begin{multline}
    \label{eq:complexity-optimized}
        \mathcal O\Bigl(
        k\Bigl(\frac{N^2}{m}\Bigr)^{\frac{1}{2K+1}}
        \bigl(kg_0t\log(1/\eps)\bigr)^{1+\frac{1}{2K+1}}
        \\
        {}
        \bigl(\Gamma+m\log(1/\eps)\bigr)
        \Bigr)
    \end{multline}
    gates.
\end{theorem}

\begin{proof}
    Set $x=t/\nu$ and take
    \begin{align}
    \label{eq:choice-nu-m}
        \nu
        =
        \Theta\Bigl(
            \Bigl(\frac{N^2}{m}\Bigr)^{\frac{1}{2K+1}}
            (q_0kg_0t)^{1+\frac{1}{2K+1}}
        \Bigr)
    \end{align}
    with a sufficiently large constant prefactor. The definition
    $\tilde g=(a_{\max}\kappa_K+1)g_0x$ gives
    \begin{align*}
        q_0k\tilde g
        &={}
        \frac{q_0k(a_{\max}\kappa_K+1)g_0t}{\nu} \\
        &={}
        \mathcal O\Bigl(
            \Bigl(
                \frac{m}{N^2q_0kg_0t}
            \Bigr)^{\frac{1}{2K+1}}
        \Bigr).
    \end{align*}
    Since $m\le N$ by \eq{choice-m} and $q_0kg_0t\ge1$, a sufficiently large
    constant prefactor in \eq{choice-nu-m} ensures
    $4eq_0k\tilde g\le e^{-1}$. For $1\le z\le2$, define
    \begin{align}
    \label{eq:lambda-single-z}
        \lambda_{\mathrm{single}}(z)
        =
        \frac{1}{m}
        \sum_{q=K+1}^{q_0}z^q\lambda_q.
    \end{align}
    Using \eq{lambda-q},
    \begin{align}
    \label{eq:lambda-single-bound}
        \lambda_{\mathrm{single}}(z)
        &\le{}
        \frac{N}{m}
        \sum_{q=K+1}^{q_0}(2zq_0k\tilde g)^q \nonumber\\
        &={}
        \mathcal O\Bigl(
            \frac{N}{m}
            (q_0k\tilde g)^{K+1}
        \Bigr) \nonumber\\
        &={}
        \mathcal O\Bigl(
            \bigl(
                Nm^K(q_0kg_0t)^{K+1}
            \bigr)^{-\frac{1}{2K+1}}
        \Bigr).
    \end{align}
    The last equality uses \eq{choice-nu-m} and
    $\tilde g=\mathcal O(g_0t/\nu)$. Squaring
    \eq{lambda-single-bound} and multiplying by $m\nu$ gives
    \begin{align}
    \label{eq:normalization-condition}
        m\nu\lambda_{\mathrm{single}}(z)^2
        &={}
        \mathcal O\Bigl(
            \frac{N^2\nu}{m}
            (q_0k\tilde g)^{2K+2}
        \Bigr) \nonumber\\
        &={}
        \mathcal O\Bigl(
            \frac{
                N^2(q_0kg_0t)^{2K+2}
            }{
                m\nu^{2K+1}
            }
        \Bigr)
        =
        \mathcal O(1),
    \end{align}
    where the second equality uses $\tilde g=\mathcal O(g_0t/\nu)$, and the last
    follows from \eq{choice-nu-m}. Since $m\nu\ge1$,
    \eq{normalization-condition} with $z=2$ implies
    $\lambda_{\mathrm{single}}(2)=\mathcal O(1)$ and hence
    $\lambda_{\mathrm{single}}=\lambda_{\mathrm{single}}(1)=\mathcal O(1)$.

    Using the definition of $\lambda_{\mathrm{multi}}$ in
    \eq{lambda-single}, we then have
    \begin{align*}
        \lambda_{\mathrm{multi}}
        =
        e^{\lambda_{\mathrm{single}}}
        -1-\lambda_{\mathrm{single}}
        =
        \mathcal O(\lambda_{\mathrm{single}}^2).
    \end{align*}
    Therefore, for a constant $C>0$, the LCQC $1$-norm of the fractional
    Trotter remainder in \eq{fractional-remainder} satisfies
    \begin{align*}
        \lambda_{\mathrm{paired}}
        \le
        1+C\lambda_{\mathrm{single}}^2
        \le
        e^{C\lambda_{\mathrm{single}}^2}.
    \end{align*}
    \algo{compensated-simulation} applies this LCQC $m$ times
    after each of the $\nu$ product-formula steps. The resulting LCQC therefore
    has $1$-norm
    \begin{align*}
        \lambda_{\mathrm{paired}}^{m\nu}
        \le
        e^{Cm\nu\lambda_{\mathrm{single}}(1)^2}
        =
        \mathcal O(1).
    \end{align*}
    By \lem{random-channel}, $\mathcal O(\eps^{-2})$ repetitions suffice to
    make the statistical error at most $\eps\|O\|/2$ with constant success
    probability.

    The choice above gives
    $8eq_0k\tilde g\le2/e<1$, so the condition of
    \lem{truncation-BCH} is satisfied. By this lemma and
    \eq{normalization-condition} with $z=1$, the accumulated BCH truncation
    error is bounded by
    \begin{align*}
        2\nu N(4eq_0k\tilde g)^{q_0+1}
        &={}
        \mathcal O\Bigl(
            \nu N(q_0k\tilde g)^{2K+2}
            (4eq_0k\tilde g)^{q_0-2K-1}
        \Bigr)\\
        &={}
        \mathcal O\Bigl(
            \frac{m}{N}
            (4eq_0k\tilde g)^{q_0-2K-1}
        \Bigr)\\
        &={}
        \mathcal O\bigl(
            (4eq_0k\tilde g)^{q_0-2K-1}
        \bigr),
    \end{align*}
    where the last line uses $m\le N$. Since
    $4eq_0k\tilde g\le e^{-1}$, this is
    $\mathcal O(e^{-(q_0-2K-1)})$.
    Choosing $q_0=\Theta(\log(1/\eps))$ with a sufficiently large constant
    makes this contribution at most $\eps/4$.

    We next bound the contribution from samples discarded in
    \lin{total-order-cutoff}. Write the full-circuit LCQC sampled by
    \algo{compensated-simulation} before the cutoff as
    $\sum_j\beta_j\mc N_j$, where each nonzero $\mc N_j$ is the composition of
    all product-formula steps and sampled compensation channels in one circuit.
    For the $j$-th term, let $d_j$ be the sum of the integers returned by the
    $m\nu$ one-step samplers:
    \begin{align*}
        d_j
        =
        \sum_{\ell=1}^{\nu}\sum_{a=1}^m s_{\ell,a}.
    \end{align*}
    In \algo{HNCC}, \textsf{HNCCSingle} returns $s=0$, whereas
    \textsf{HNCCMulti} returns $s=q_1+\cdots+q_r$ for a term
    $\prod_{u=1}^r\Phi_{q_u}(x)$. Hence weighting each coefficient by
    $2^{d_j}$ leaves the \textsf{HNCCSingle} contribution unchanged and
    replaces each $\lambda_q$ by $2^q\lambda_q$ in the \textsf{HNCCMulti}
    contribution. Since the $m\nu$ fractional Trotter remainders are sampled
    independently,
    \begin{align*}
        &\sum_j2^{d_j}|\beta_j|\\
        ={}&\biggl(1+\frac{\lambda_{\mathrm{single}}^2}{2}+
            \sum_{r=2}^{\infty}\frac{1}{r!}
            \biggl(
                \frac{1}{m}\sum_{q=K+1}^{q_0}2^q\lambda_q
            \biggr)^r\biggr)^{m\nu}\\
        ={}&
        \biggl(
            1+\frac{\lambda_{\mathrm{single}}(1)^2}{2}
            +\sum_{r=2}^{\infty}
            \frac{\lambda_{\mathrm{single}}(2)^r}{r!}
        \biggr)^{m\nu}\\
        \le{}&
        e^{Cm\nu\lambda_{\mathrm{single}}(2)^2}
        =
        \mathcal O(1),
    \end{align*}
    where the inequality uses
    $\lambda_{\mathrm{single}}(1)\le
    \lambda_{\mathrm{single}}(2)=\mathcal O(1)$. The last equality follows
    from \eq{normalization-condition} with $z=2$.

    The cutoff in \lin{total-order-cutoff} replaces the terms with $d_j>s_0$
    by the zero map. Since $\|\mc N_j\|_\diamond\le1$, the omitted contribution
    satisfies
    \begin{align}
    \label{eq:total-order-cutoff-bound}
        \Bigl\|
            \sum_{j:d_j>s_0}\beta_j\mc N_j
        \Bigr\|_\diamond
        &\le
        \sum_{j:d_j>s_0}|\beta_j|\\
        &\le
        2^{-s_0}
        \sum_j 2^{d_j}|\beta_j|.
    \end{align}
    Together with the weighted LCQC $1$-norm bound above, this is
    $\mathcal O(2^{-s_0})$.
    The resulting bias in the observable expectation is therefore at most
    $\mathcal O(2^{-s_0})\|O\|$. Choosing
    $s_0=\Theta(\log(1/\eps))$ with a sufficiently large constant makes this
    bias at most $\eps\|O\|/4$.

    There are $m\nu$ calls to \algo{HNCC}. By \thm{hncc}, an output with
    $s_{\ell,a}=0$ uses $\mathcal O(kq_0)$ gates, whereas an output with
    $s_{\ell,a}>0$ uses $\mathcal O(ks_{\ell,a})$ gates. Summing these bounds,
    the compensation channels retained by \algo{compensated-simulation} use
    \begin{align*}
        \mathcal O\Bigl(
            \nu m kq_0
            +
            k\sum_{\ell=1}^{\nu}\sum_{a=1}^m s_{\ell,a}
        \Bigr)
        =
        \mathcal O(\nu m kq_0+ks_0)
    \end{align*}
    gates, where the second expression uses the cutoff condition
    $\sum_{\ell=1}^{\nu}\sum_{a=1}^m s_{\ell,a}\le s_0$.
    The product-formula steps use
    $\mathcal O(\nu\Gamma k)$ gates. Hence every circuit uses
    \begin{align*}
        \mathcal O\bigl(
            \nu k(\Gamma+mq_0)
            +
            ks_0
        \bigr)
    \end{align*}
    gates. The choices above satisfy $s_0=\mathcal O(q_0)$, and
    $m\nu\ge1$. Hence $ks_0=\mathcal O(\nu m kq_0)$, so the gate count is
    $\mathcal O(\nu k(\Gamma+mq_0))$. Substituting
    \eq{choice-nu-m} and
    $q_0=\mathcal O(\log(1/\eps))$ gives
    \eq{complexity-optimized}. The factor
    $m^{-1/(2K+1)}(\Gamma+m\log(1/\eps))$
    is minimized up to a constant factor by \eq{choice-m}.
\end{proof}


We also consider unpaired HNCC, which represents the part linear in the BCH
generator using the parameter-shift rule. We use
\algo{compensated-simulation} with $m=1$, replace each call to \algo{HNCC} by
the sampler in
\thm{unpaired-hncc}, and replace $\lambda_{\mathrm{paired}}$ by
$\lambda_{\mathrm{unpaired}}$. The same
cutoff on the total BCH order is applied to the returned integers.

\begin{theorem}[Unpaired high-order nested-commutator compensation]
\label{thm:unpaired-global-order}
    Under the assumptions of \thm{complexity}, unpaired HNCC estimates
    $\tr[O\mc U(t)(\rho_0)]$ to precision $\eps\|O\|$ using
    $\mathcal O(\eps^{-2})$ circuit repetitions. Every circuit uses
    \begin{align*}
        \mathcal O\bigl(
            \Gamma k
            N^{1/K}
            \bigl(
                kg_0t\log(1/\eps)
            \bigr)^{1+1/K}
        \bigr)
    \end{align*}
    gates. The compensation part of each circuit uses
    $\mathcal O(k\log(1/\eps))$ Clifford gates and no $T$ gates.
\end{theorem}

\begin{proof}
    Set $q_0=\Theta(\log(1/\eps))$, $x=t/\nu$, and
    \begin{align}
    \label{eq:choice-nu-unpaired}
        \nu
        =
        \Theta\bigl(
            N^{1/K}
            (q_0kg_0t)^{1+1/K}
        \bigr)
    \end{align}
    with a sufficiently large constant prefactor. Since
    $\tilde g=\mathcal O(g_0t/\nu)$, this choice gives
    \begin{align*}
        q_0k\tilde g
        =
        \mathcal O\bigl(
            (Nq_0kg_0t)^{-1/K}
        \bigr).
    \end{align*}
    We may therefore choose the prefactor in
    \eq{choice-nu-unpaired} so that
    $4eq_0k\tilde g\le e^{-1}$. For $1\le z\le2$, take $m=1$ in
    \eq{lambda-single-z}. The bound \eq{lambda-single-bound} gives
    \begin{align*}
        \lambda_{\mathrm{single}}(z)
        =
        \mathcal O\bigl(N(q_0k\tilde g)^{K+1}\bigr).
    \end{align*}
    It follows that
    \begin{align}
    \label{eq:normalization-unpaired}
        \nu\lambda_{\mathrm{single}}(z)
        &={}
        \mathcal O\bigl(
            \nu N(q_0k\tilde g)^{K+1}
        \bigr)\nonumber\\
        &={}
        \mathcal O\Bigl(
            \frac{
                N(q_0kg_0t)^{K+1}
            }{
                \nu^K
            }
        \Bigr)
        =
        \mathcal O(1),
    \end{align}
    where the last equality follows from \eq{choice-nu-unpaired}.

    By \thm{unpaired-hncc}, the LCQC $1$-norm over all $\nu$ steps is
    \begin{align*}
        \lambda_{\mathrm{unpaired}}^{\nu}
        =
        e^{\nu\lambda_{\mathrm{single}}(1)}
        =
        \mathcal O(1).
    \end{align*}
    Hence \lem{random-channel} shows that
    $\mathcal O(\eps^{-2})$ repetitions suffice to make the statistical error
    at most $\eps\|O\|/2$ with constant success probability.

    The choice above gives $8eq_0k\tilde g\le2/e<1$, so
    \lem{truncation-BCH} applies. Using
    \eq{normalization-unpaired} with $z=1$, the accumulated BCH truncation
    error is bounded by
    \begin{align*}
        2\nu N(4eq_0k\tilde g)^{q_0+1}
        =
        \mathcal O\bigl(
            (4eq_0k\tilde g)^{q_0-K}
        \bigr).
    \end{align*}
    Since $4eq_0k\tilde g\le e^{-1}$, a sufficiently large constant in the
    choice of $q_0$ makes this error at most $\eps/4$.

    Write the full-circuit LCQC before the cutoff on the total BCH order as
    $\sum_j\beta_j\mc N_j$. For the $j$-th term, let
    $d_j=\sum_{\ell=1}^{\nu}s_\ell$, where $s_\ell$ is the integer returned by
    the unpaired sampler after the $\ell$-th product-formula step. In
    \thm{unpaired-hncc}, the identity returns $s_\ell=0$, while a term with
    BCH orders $q_1,\ldots,q_r$ returns
    $s_\ell=q_1+\cdots+q_r$. Therefore,
    \begin{align*}
        \sum_j2^{d_j}|\beta_j|
        &={}
        \biggl[
            1+
            \sum_{r=1}^{\infty}\frac{1}{r!}
            \biggl(
                \sum_{q=K+1}^{q_0}2^q\lambda_q
            \biggr)^r
        \biggr]^{\nu}\\
        &={}
        e^{\nu\lambda_{\mathrm{single}}(2)}
        =
        \mathcal O(1),
    \end{align*}
    where the last equality follows from
    \eq{normalization-unpaired} with $z=2$.

    The cutoff replaces terms with $d_j>s_0$ by the zero map. Applying
    \eq{total-order-cutoff-bound} and the weighted LCQC $1$-norm bound above,
    the omitted contribution has diamond norm $\mathcal O(2^{-s_0})$.
    Taking $s_0=\Theta(\log(1/\eps))$ with a sufficiently large constant
    makes the resulting observable bias at most $\eps\|O\|/4$.

    By \thm{unpaired-hncc}, a retained sample uses
    \(
        \mathcal O(k\sum_{\ell=1}^{\nu}s_\ell)
        =
        \mathcal O(ks_0)
    \)
    compensation gates, all of which are Clifford gates. The product-formula
    steps use $\mathcal O(\nu\Gamma k)$ gates. Since
    $s_0=\mathcal O(q_0)$, $N,\Gamma\ge1$, and $kg_0t\ge1$,
    \eq{choice-nu-unpaired} gives
    $\nu=\Omega(q_0)$. Hence the product-formula cost absorbs the compensation
    cost. Substituting \eq{choice-nu-unpaired} and
    $q_0=\Theta(K+\log(1/\eps))$ gives the stated gate count and completes
    the proof.
\end{proof}

\section{One-step sampler implementation}
\label{sec:implementation}

This section proves \thm{hncc} and \thm{unpaired-hncc}. The top-level paired
sampler chooses between the part linear in the BCH generator and products of
multiple BCH terms. Both cases call
\textsf{BCHSample}, which in turn calls
\textsf{LCSample}. The calling relation among these routines is shown in
\fig{hncc-flow}.

\begin{algorithm}[H]
\caption{One-step High-order Nested-commutator Compensation (\textsf{HNCC})}
\label{algo:HNCC}
\KwInput{Product-formula data $\mathsf{PFData}$, time step $x$, truncation
order $q_0\ge K+1$, and number of factors $m$.}
\KwOutput{Sign $\eta$, superoperator $\mc N$, and nonnegative integer $s$.}

Form the scaled Pauli-term list
$\mathsf{TermData}\gets
\{(\tilde{\alpha}_{v,\gamma},\tilde P_{v,\gamma})\}_{v\in[M],\gamma}$
according to \eq{def-tilde-H} and \eq{decompose-tilde-H}\;
Calculate
$\tilde g\gets(a_{\max}\kappa_K+1)g_0x$\;
Set
$\mathsf{CommData}\gets
(\mathsf{TermData},M,N,k,\tilde g)$\;

\ForEach{$q\in[K+1,q_0]$}{
    $\lambda_q\gets
    q!(2k\tilde g)^qk^{-1}N/q^2$\;
}
Set $\lambda_{\mathrm{single}}$, $\lambda_{\mathrm{multi}}$, and
$\lambda_{\mathrm{paired}}$ according to \eq{lambda-single}\;

\WithProb{$\frac{1+\lambda_{\mathrm{single}}^2/2}{\lambda_{\mathrm{paired}}}$}{
    $(\eta,\mc N)\gets\textsf{HNCCSingle}(\mathsf{CommData},\{\lambda_q\},\lambda_{\mathrm{single}},m)$\;
    \Return $(\eta,\mc N,0)$\;
}
\Else{
    \Return $\textsf{HNCCMulti}(\mathsf{CommData},\{\lambda_q\},\lambda_{\mathrm{single}},m)$\;
}
\end{algorithm}

\begin{figure}[tbp]
    \centering
    \includegraphics[width=0.47\textwidth]{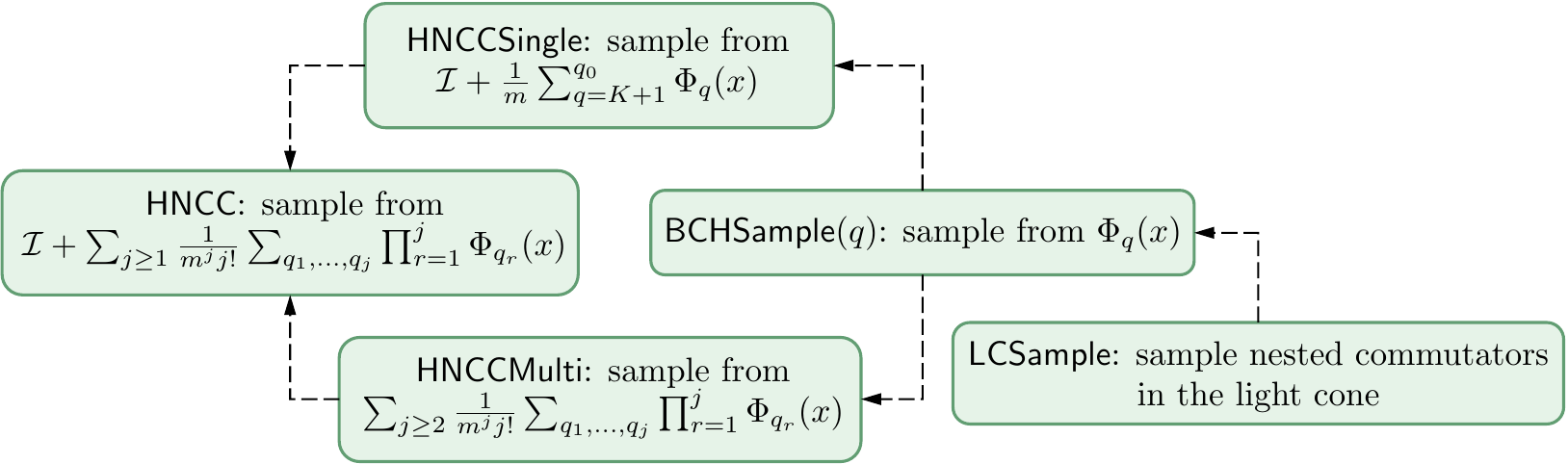}
    \caption{Calling relation among the routines used for one \textsf{HNCC}
    sample.}
    \label{fig:hncc-flow}
\end{figure}

\subsection{Products of BCH terms}

For these terms, we sample $j\ge2$ with
probability
\begin{align}
\label{eq:prob-j}
    \frac{
        \lambda_{\mathrm{single}}^j
    }{
        j!\lambda_{\mathrm{multi}}
    },
\end{align}
and, conditioned on $j$, sample $q_1,\ldots,q_j$ independently with
probability $\lambda_q/(m\lambda_{\mathrm{single}})$.

\begin{algorithm}[htbp]
\caption{Sampling Products of BCH Terms (\textsf{HNCCMulti})}
\label{algo:HNCC-multi}
\KwInput{$\mathsf{CommData}$, $\{\lambda_q\}_{q=K+1}^{q_0}$,
$\lambda_{\mathrm{single}}$, and $m$.}
\KwOutput{Sign $\eta$, superoperator $\mc N$, and nonnegative integer $s$.}

Initialize $V\gets I$ and $\eta\gets1$\;
Set $\lambda_{\mathrm{multi}}
\gets\sum_{j=2}^{\infty}\lambda_{\mathrm{single}}^j/j!$\;
Sample $j\ge2$ according to \eq{prob-j}\;
Independently sample $q_1,\ldots,q_j$ with probabilities
$\lambda_q/(m\lambda_{\mathrm{single}})$\;
Set $s\gets q_1+\cdots+q_j$\;

\For{$r=1$ \KwTo $j$}{
    $(\tilde\eta_r,C_r)
    \gets
    \mathsf{BCHSample}(\mathsf{CommData},q_r)$\;
    \lIf{$C_r=0$}{\Return $(1,0,s)$}
    \WithProb{$1/2$}{
        $V_r\gets e^{\frac{\pi}{4}C_r}$\;
        $\eta_r\gets\tilde\eta_r$\;
    }
    \Else{
        $V_r\gets e^{-\frac{\pi}{4}C_r}$\;
        $\eta_r\gets-\tilde\eta_r$\;
    }
    Update $\eta\gets\eta_r\eta$ and $V\gets V_rV$\;
}
\Return $(\eta,\mc U_V,s)$\;
\end{algorithm}

\begin{lemma}
\label{lem:hncc-multi}
    Under the assumptions of \thm{hncc},
    \algo{HNCC-multi} samples $(\eta,\mc N,s)$ such that
    \begin{align*}
        \E[\eta\mc N]
        =
        \frac{1}{\lambda_{\mathrm{multi}}}
        \sum_{j=2}^{\infty}
        \frac{1}{m^jj!}
        \sum_{q_1,\ldots,q_j=K+1}^{q_0}
        \prod_{r=1}^j\Phi_{q_r}(x).
    \end{align*}
\end{lemma}

\begin{proof}
    Fix $(j,q_1,\ldots,q_j)$. For each $r$, averaging over the final random
    choice in
    \algo{HNCC-multi} gives the signed channel
    $\tilde\eta_r(\mc U_{e^{\pi C_r/4}}
    -\mc U_{e^{-\pi C_r/4}})/2$. By the parameter-shift identity and
    \lem{BCH-sample},
    \begin{align*}
        \E\Bigl[
            \tilde\eta_r\frac{1}{2}\ad_{C_r}
            \Bigm|q_r
        \Bigr]
        =
        \frac{1}{\lambda_{q_r}}\Phi_{q_r}(x),
    \end{align*}
    where the equality also holds when $C_r=0$, since the algorithm then
    returns the zero map. The calls to \textsf{BCHSample} are independent, so
    \begin{align*}
        \E[\eta\mc N\mid j,q_1,\ldots,q_j]
        =
        \biggl(
            \prod_{r=1}^j\frac{1}{\lambda_{q_r}}
        \biggr)
        \prod_{r=1}^j\Phi_{q_r}(x).
    \end{align*}
    The joint probability of the sampled tuple is
    \begin{align*}
        \frac{\lambda_{\mathrm{single}}^j}{
            j!\lambda_{\mathrm{multi}}
        }
        \prod_{r=1}^j
        \frac{\lambda_{q_r}}{
            m\lambda_{\mathrm{single}}
        }
        =
        \frac{1}{
            m^jj!\lambda_{\mathrm{multi}}
        }
        \prod_{r=1}^j\lambda_{q_r}.
    \end{align*}
    Averaging the conditional expectation with this probability proves the
    claim.
\end{proof}

\subsection{Terms linear in the BCH generator}

By \lem{BCH-sample}, every nonzero sampled operator satisfies
$C\in\pm\i\mathcal P_N$ and hence $C^2=-I$. This allows the identity
contribution to be paired with the sampled adjoint action. Let
$\theta=\tan^{-1}(\lambda_{\mathrm{single}}/2)$.

\begin{algorithm}[htbp]
\caption{Sampling Terms Linear in the BCH Generator
(\textsf{HNCCSingle})}
\label{algo:HNCC-single}
\KwInput{$\mathsf{CommData}$, $\{\lambda_q\}_{q=K+1}^{q_0}$,
$\lambda_{\mathrm{single}}$, and $m$.}
\KwOutput{Sign $\eta$ and superoperator $\mc N$.}

$\theta\gets\tan^{-1}(\lambda_{\mathrm{single}}/2)$\;
Sample $q\in\{K+1,\ldots,q_0\}$ with probability
$\lambda_q/(m\lambda_{\mathrm{single}})$\;
$(\tilde\eta,C)\gets
\mathsf{BCHSample}(\mathsf{CommData},q)$\;

\If{$C=0$}{
    \lWithProb{$\frac{\cos^2\theta}{1+\sin^2\theta}$}{
        \Return $(1,\mc I)$
    }
    \lElse{\Return $(1,0)$}
}
\WithProb{$\frac{1}{1+\sin^2\theta}$}{
    \Return $(1,\mc U_{e^{\theta\tilde\eta C}})$
}
\lElse{
    \Return $(-1,\mc U_C)$
}
\end{algorithm}

\begin{lemma}
\label{lem:hncc-single}
    Under the assumptions of \thm{hncc},
    \algo{HNCC-single} samples a sign $\eta$ and a superoperator $\mc N$
    such that
    \begin{align*}
        \E[\eta\mc N]
        =
        \frac{1}{
            1+\lambda_{\mathrm{single}}^2/2
        }
        \biggl(
            \mc I+
            \frac{1}{m}
            \sum_{q=K+1}^{q_0}\Phi_q(x)
        \biggr).
    \end{align*}
\end{lemma}

\begin{proof}
    Fix the sampled order $q$ and condition on $(\tilde\eta,C)$. If
    $C\ne0$, averaging over the final random choice in \algo{HNCC-single} and using
    the identity-pairing formula gives
    \begin{align*}
        \frac{\cos^2\theta}{1+\sin^2\theta}
        \biggl(
            \mc I+
            \tilde\eta\frac{\lambda_{\mathrm{single}}}{2}\ad_C
        \biggr).
    \end{align*}
    The same expression holds for $C=0$: the algorithm returns $\mc I$ with
    probability $\cos^2\theta/(1+\sin^2\theta)$ and the zero map otherwise.
    Averaging over $(\tilde\eta,C)$ and applying \lem{BCH-sample} therefore
    gives
    \begin{align*}
        \E[\eta\mc N\mid q]
        =
        \frac{\cos^2\theta}{1+\sin^2\theta}
        \Bigl(
            \mc I+
            \frac{\lambda_{\mathrm{single}}}{\lambda_q}
            \Phi_q(x)
        \Bigr).
    \end{align*}
    Averaging with probability
    $\lambda_q/(m\lambda_{\mathrm{single}})$ gives
    \begin{align*}
        \E[\eta\mc N]
        &=
        \frac{\cos^2\theta}{1+\sin^2\theta}
        \biggl(
            \mc I+
            \frac{1}{m}
            \sum_{q=K+1}^{q_0}\Phi_q(x)
        \biggr)\\
        &=
        \frac{1}{
            1+\lambda_{\mathrm{single}}^2/2
        }
        \biggl(
            \mc I+
            \frac{1}{m}
            \sum_{q=K+1}^{q_0}\Phi_q(x)
        \biggr).
    \end{align*}
\end{proof}

Combining \lem{hncc-single} and \lem{hncc-multi} gives the one-step sampler
in \thm{hncc}.

\begin{proof}[Proof of \thm{hncc}]
    The algorithm calls \algo{HNCC-single} with probability
    $(1+\lambda_{\mathrm{single}}^2/2)/\lambda_{\mathrm{paired}}$ and
    \algo{HNCC-multi} with probability
    $\lambda_{\mathrm{multi}}/\lambda_{\mathrm{paired}}$.
    The two lemmas give
    \begin{align*}
        \E[\eta\mc N]
        &=
        \frac{1}{\lambda_{\mathrm{paired}}}
        \sum_{j=0}^{\infty}
        \frac{1}{m^jj!}
        \biggl(
            \sum_{q=K+1}^{q_0}\Phi_q(x)
        \biggr)^j\\
        &=
        \frac{1}{\lambda_{\mathrm{paired}}}
        \exp\biggl(
            \frac{1}{m}
            \sum_{q=K+1}^{q_0}\Phi_q(x)
        \biggr).
    \end{align*}
    We next prove the gate bounds. By \lem{lcsample}, a nonzero order-$q$ call
    to \textsf{BCHSample} returns a Pauli operator generated by a nested
    commutator of $q$ Pauli operators, each supported on at most $k$ qubits.
    Its support is contained in the union of these $q$ supports and hence has
    size at most $kq$. A Pauli rotation with this generator can be implemented
    using $\mathcal O(kq)$ gates. Sampling a nonzero term from the linear part
    calls
    \textsf{BCHSample} once with $q\le q_0$ and returns $s=0$, so the resulting
    channel uses $\mathcal O(kq_0)$ gates. For products of multiple BCH terms,
    \algo{HNCC-multi} sets $s=q_1+\cdots+q_j$ and composes $j$ such rotations,
    which use $\mathcal O(k\sum_{r=1}^j q_r)=\mathcal O(ks)$ gates in total.
    A zero-map output requires no gates.

    Conditioned on calling \algo{HNCC-multi}, the expected number
    of calls to \textsf{BCHSample} is
    \begin{align*}
        \E[j]
        &={}
        \frac{1}{\lambda_{\mathrm{multi}}}
        \sum_{j=2}^{\infty}
        \frac{j\lambda_{\mathrm{single}}^j}{j!}\\
        &={}
        \frac{
            \lambda_{\mathrm{single}}
            (e^{\lambda_{\mathrm{single}}}-1)
        }{
            e^{\lambda_{\mathrm{single}}}
            -1-\lambda_{\mathrm{single}}
        }=
        \mathcal O(1)
    \end{align*}
    when $\lambda_{\mathrm{single}}$ is bounded by a constant. Here we used
    $e^x-1-x\ge x^2/2$ and $e^x-1=\mathcal O(x)$ for bounded $x\ge0$.
    By \lem{BCH-sample}, each call has
    cost polynomial in $N$, $\Gamma$, and $q_0$. This proves the
    expected classical-cost bound.
\end{proof}

The same BCH-term sampler gives the unpaired construction by treating every
term in the Taylor expansion through the parameter-shift representation.

\begin{proof}[Proof of \thm{unpaired-hncc}]
    Sample an integer $j\ge0$ with probability
    \(
        {\lambda_{\mathrm{single}}^j}/({
            j!\lambda_{\mathrm{unpaired}}}).
    \)
    For $j=0$, return $(1,\mc I,0)$. Otherwise, independently sample
    $q_1,\ldots,q_j$ with probabilities
    $\lambda_q/(m\lambda_{\mathrm{single}})$ and sample the corresponding
    BCH terms exactly as in \algo{HNCC-multi}. Set
    $s=q_1+\cdots+q_j$.

    For fixed $j$ and $q_1,\ldots,q_j$, the parameter-shift identity and
    \lem{BCH-sample} give
    \begin{align*}
        \E[\eta\mc N\mid j,q_1,\ldots,q_j]
        =
        \biggl(
            \prod_{r=1}^j\frac{1}{\lambda_{q_r}}
        \biggr)
        \prod_{r=1}^j\Phi_{q_r}(x).
    \end{align*}
    The joint probability of these orders is
    \begin{align*}
        \frac{\lambda_{\mathrm{single}}^j}{
            j!\lambda_{\mathrm{unpaired}}}
        \prod_{r=1}^j
        \frac{\lambda_{q_r}}{m\lambda_{\mathrm{single}}}
        =
        \frac{1}{m^jj!\lambda_{\mathrm{unpaired}}}
        \prod_{r=1}^j\lambda_{q_r}.
    \end{align*}
    Including the identity term and averaging over all orders therefore
    gives
    \begin{align*}
        \E[\eta\mc N]
        &={}
        \frac{1}{\lambda_{\mathrm{unpaired}}}
        \sum_{j=0}^{\infty}
        \frac{1}{m^jj!}
        \biggl(
            \sum_{q=K+1}^{q_0}\Phi_q(x)
        \biggr)^j\\
        &={}
        \frac{1}{\lambda_{\mathrm{unpaired}}}
        \exp\biggl(
            \frac{1}{m}\sum_{q=K+1}^{q_0}\Phi_q(x)
        \biggr).
    \end{align*}

    The support argument in the proof of \thm{hncc} shows that the sampled
    order-$q_r$ Pauli rotation uses $\mathcal O(kq_r)$ gates. Hence every
    nonidentity output uses
    $\mathcal O(k\sum_{r=1}^j q_r)=\mathcal O(ks)$ gates. These are
    $\pi/4$ Pauli rotations and therefore Clifford gates. Finally,
    \begin{align*}
        \E[j]
        =
        \frac{1}{\lambda_{\mathrm{unpaired}}}
        \sum_{j=1}^{\infty}
        \frac{j\lambda_{\mathrm{single}}^j}{j!}
        =
        \lambda_{\mathrm{single}}.
    \end{align*}
    Thus, when $\lambda_{\mathrm{single}}$ is bounded by a constant, the
    expected number of BCH samples is $\mathcal O(1)$. The cost bound in
    \lem{BCH-sample} then proves the expected classical-cost bound.
\end{proof}

\subsection{BCH-term sampling}
We describe a sampler for a normalized Pauli commutator whose expected adjoint
action matches the BCH term of a general set of $k$-local Hamiltonians $\{H_v\}_{v=1}^M$ with Pauli decompositions $H_v=\sum_{\gamma}\alpha_{v, \gamma} P_{v, \gamma}$ such that $\alpha_{v,\gamma}\in\mathbb{R}$ and $\{\alpha_{v,\gamma}P_{v,\gamma}\}_{v,\gamma}$ is $\tilde g$-extensive. Identity Pauli terms are omitted, since they vanish under commutators. Define $h_{v,\gamma} = \alpha_{v,\gamma}P_{v,\gamma}$ such that $H_v = \sum_{\gamma} h_{v,\gamma}$. 

The BCH term $\Phi_q$ of order $q$ is given by
\begin{widetext}
\begin{align}
    \Phi_q &= \sum_{\substack{p_v \ge 0,\\p_1+\cdots+p_{M} = q}}\frac{1}{\prod_{v=1}^{M} p_v!} \sum_{\sigma\in S_q} c_{\sigma, q}\mc R_{\sigma}\{[\underbrace{-\i \ad_{{H}_{M}}, \ldots, -\i \ad_{{H}_{M}}}_{p_{M}}, \ldots, \underbrace{-\i \ad_{{H}_{1}}, \ldots, -\i \ad_{{H}_{1}}}_{p_1}]\}\nonumber\\
    &=\sum_{\substack{p_v \ge 0,\\p_1+\cdots+p_{M} = q}}\frac{(- 1)^{q} }{\prod_{v=1}^{M} p_v!} \sum_{\sigma\in S_q} c_{\sigma, q}\sum_{w_{1,1}, \ldots, w_{M,p_M}}\ad_{\mc R_{\sigma}\{B_{\vec{p},\vec{w}}\}}, \label{eq:BCH-expression}
\end{align}
where the ordered commutator $B_{\vec{p},\vec{w}}$ is defined as 
\begin{align*}
    B_{\vec{p}, \vec{w}} := \big[ \underbrace{\i {h}_{M, w_{M, 1}}, \ldots, \i {h}_{M, w_{M, p_M}}}_{p_{M}}, \ldots, \underbrace{\i {h}_{1, w_{1, 1}}, \ldots, \i {h}_{1, w_{1, p_1}}}_{p_1} \big].
\end{align*} 
\end{widetext}
The second equality in \eq{BCH-expression} follows from substituting the Pauli decomposition of ${H}_v$ and applying the Jacobi identity. The index vectors $\vec{p}=(p_1, \ldots, p_M)$ and $\vec{w} = (w_{1, 1}, \ldots, w_{M, p_M})$ denote the multiplicity of each ${H}_v$ and the specific Pauli indices in the decomposition of ${H}_v$, respectively.

Rewrite the sum in $\Phi_q$ in a form suitable for sampling. Denote the commutator $\mc{R}_{\sigma}\{B_{\vec{p}, \vec{w}}\}$ by 
\begin{align}
\label{eq:permutation-equation}
    C_{\vec{v}, \vec{\gamma}}=  [\i h_{v_q, \gamma_q}, \ldots, \i h_{v_1, \gamma_1}]=\mc{R}_{\sigma}\{B_{\vec{p}, \vec{w}}\},
\end{align}
where $\vec{v} = (v_1, \ldots, v_q)$ and $\vec{\gamma} = (\gamma_1, \ldots, \gamma_q)$ are the indices of the terms in the commutator after permutation. 
Given the multiplicity vector $\vec{p}$ and the index vector $\vec{w}$, the permutation $\sigma$ uniquely determines $\vec{v}$ and $\vec{\gamma}$. Conversely, given a sequence $\vec{v}$, the multiplicity $\vec{p}$ of the Hamiltonians $\{{H}_v\}_{v=1}^M$ is uniquely determined. For each $v\in[M]$, let
\begin{align*}
    B_v
    =
    \biggl\{
        \sum_{u=1}^{v-1}p_u+1,\ldots,
        \sum_{u=1}^{v}p_u
    \biggr\}
\end{align*}
be the positions occupied by terms from $H_v$ before applying
$\mc R_\sigma$. Since the term at position $r$ after reordering comes from
position $\sigma(r)$, the consistent permutations are
\begin{align}
\label{eq:def-sigma}
    \Sigma_{\vec v}
    =
    \{\sigma\in S_q:\sigma(r)\in B_{v_r},\ \forall r\in[q]\}.
\end{align}
The cardinality of this set is precisely $|\Sigma_{\vec{v}}| = \prod_{v=1}^M p_v!$.
Given $\vec v$, $\vec\gamma$, and $\sigma\in\Sigma_{\vec v}$, the original
Pauli indices are recovered from
\begin{align*}
    w_{v,j}
    =
    \gamma_{\sigma^{-1}(\sum_{u=1}^{v-1}p_u+j)}.
\end{align*}
Consequently, the summation over $\vec{p}$ and $\sigma \in S_q$ in \eqref{eq:BCH-expression} can be reindexed as a summation over all possible sequences $\vec{v} \in [M]^q$ and their respective consistent permutations $\sigma \in \Sigma_{\vec{v}}$:
\begin{align}
\label{eq:BCH-term}
    \Phi_q = \sum_{v_1, \ldots, v_{q}=1}^{M} \frac{(-1)^q}{|\Sigma_{\vec{v}}|} \sum_{\sigma \in \Sigma_{\vec{v}}} c_{\sigma, q} \sum_{\gamma_1, \ldots, \gamma_q} \ad_{C_{\vec{v}, \vec{\gamma}}}.
\end{align}

\algo{BCH-sampling} samples a BCH term. For this routine and the light-cone
sampler below, write
\[
\begin{aligned}
    \mathsf{CommData}:=\bigl(
    \{(\alpha_{v,\gamma},P_{v,\gamma})\}_{v\in[M],\gamma},M,N,k,\tilde g
    \bigr).
\end{aligned}
\]
The grouped index $v$ records the exponential factor from which a Pauli term comes. The sampling routines use this term list directly and do not require the aggregated Hamiltonians as input.
Our goal is to sample the normalized commutator $\operatorname{sgn}(c_{\sigma, q})(- 1)^{q} \frac{C_{\vec{v}, \vec{\gamma}}}{\|C_{\vec{v}, \vec{\gamma}}\|}$ with probability
\begin{align}
\label{eq:prob-comm}
    \frac{\|C_{\vec{v}, \vec{\gamma}}\|}{ q!(2k\tilde g)^{q-1}N\tilde g}\frac{q^2|c_{\sigma,q}|}{\prod_{v=1}^{M} p_v!}.
\end{align}
The algorithm generates samples according to \eq{prob-comm} by first drawing a nested commutator with the light-cone sampler (\algo{LC_sampling}) and then sampling a consistent permutation. The final rejection step is valid because \eq{def-c-sigma} implies $q^2|c_{\sigma,q}|\le 1$.

\begin{algorithm}[htbp]
\caption{Sampling for a BCH term (\textsf{BCHSample})}
\label{algo:BCH-sampling}
\KwInput{Commutator-sampling data $\mathsf{CommData}$, order $q$.}
\KwOutput{Sign $\tilde{\eta}$ and operator $C$.}
$res \gets \textsf{LCSample}(\mathsf{CommData},q)$\label{lin:lcs}\;
Parse $res$ as index vector $\vec{v}=(v_1, \ldots, v_q)$, operator $C$\;
\lIf{$C = 0$}{\Return $(1,0)$}
For each $v \in [M]$, $p_v \gets |\{ u \in [q] : v_u = v \}|$\;
Sample a permutation $\sigma \in S_q$ uniformly from $\Sigma_{\vec{v}}$ defined in \eq{def-sigma} \label{lin:sample-permutation}\;
$\tilde{\eta}\gets \operatorname{sgn}(c_{\sigma,q}) (-1)^q$\;
With probability $1-q^2|c_{\sigma, q}|$ \Return $(1, 0)$\label{lin:sample-adP}\;
\Return $(\tilde{\eta}, C)$\;
\end{algorithm}

\begin{lemma}
\label{lem:BCH-sample}
    Under the assumptions above, for any integer $q\ge1$,
$\textsf{BCHSample}(\mathsf{CommData},q)$ samples a sign $\tilde{\eta}\in \{-1,1\}$ and an operator $C\in \pm\i\mathcal{P}_N\cup\{0\}$ such that
    \begin{align*}
         &\mathbb{E}\Bigl[\tilde{\eta} \frac{1}{2}\ad_{C}\Bigr] = \frac{q^2}{q!(2k\tilde g)^{q}k^{-1}N}\Phi_{q}=\frac{1}{\lambda_q}\Phi_q, \\
         &\lambda_q:=\frac{q!(2k\tilde g)^qk^{-1}N}{q^2}.
    \end{align*}
    Its classical running time is polynomial
    in $N$, $q$, and the length of the Pauli-term list in $\mathsf{CommData}$.
\end{lemma}

\begin{proof}
    By \lem{lcsample}, if the output of $\textsf{LCSample}$ in \lin{lcs} corresponds to the sampled indices $(\vec{v}, \vec{\gamma})$ such that $C_{\vec{v}, \vec{\gamma}} \neq 0$, it returns the normalized operator
    \begin{align*}
        C = \frac{[\i h_{v_q, \gamma_q}, \ldots, \i h_{v_1, \gamma_1}]}{\|[\i h_{v_q, \gamma_q}, \ldots, \i h_{v_1, \gamma_1}]\|} = \frac{C_{\vec{v}, \vec{\gamma}}}{\|C_{\vec{v}, \vec{\gamma}}\|}
    \end{align*}
    with joint probability $p(\vec{v}, \vec{\gamma}) = \frac{\|C_{\vec{v}, \vec{\gamma}}\|}{ q!(2k\tilde g)^{q-1}N\tilde g}$. 

    In \lin{sample-permutation}, a permutation $\sigma$ is sampled uniformly from $\Sigma_{\vec{v}}$ with probability $1/|\Sigma_{\vec{v}}|$ and then accepted with probability $q^2|c_{\sigma, q}|$ to form the final valid sample. The algorithm assigns the sign $\tilde{\eta} = \operatorname{sgn}(c_{\sigma, q})(-1)^{q}$.
    
    Taking the expectation over all random choices of the path sequence $(\vec{v}, \vec{\gamma})$ and the permutation $\sigma$, and noting that samples with $C=0$ contribute $0$ to the expectation, we have
    \begin{align}
        &\mathbb{E}_{\vec{v}, \vec{\gamma}, \sigma}\Bigl[\tilde{\eta} \frac{1}{2}\ad_{C}\Bigr] \nonumber\\
		={}& \sum_{\vec{v}, \vec{\gamma}} p(\vec{v}, \vec{\gamma})
		\sum_{\sigma\in \Sigma_{\vec{v}}}
		\frac{q^2|c_{\sigma, q}|}{|\Sigma_{\vec{v}}|}
		\operatorname{sgn}(c_{\sigma, q})(-1)^{q}
        \frac{1}{2}\ad_{C_{\vec{v}, \vec{\gamma}}/\|C_{\vec{v}, \vec{\gamma}}\|} \nonumber\\
		={}& \begin{multlined}
        \frac{q^2}{q!(2k\tilde g)^{q-1}(2\tilde gN)}
        \\\sum_{v_1, \ldots,v_{q}=1}^M \frac{(-1)^{q}}{|\Sigma_{\vec{v}}|}
        \sum_{\sigma\in \Sigma_{\vec{v}}} c_{\sigma, q}  \sum_{\gamma_1, \ldots, \gamma_q}
        \ad_{C_{\vec{v}, \vec{\gamma}}}. \label{eq:expect-BCH}
        \end{multlined}
    \end{align}
    Substituting the expression for $\Phi_q$ from Eq.~\eqref{eq:BCH-term} into Eq.~\eqref{eq:expect-BCH}, we conclude that
    \begin{align*}
        \mathbb{E}\Bigl[\tilde{\eta} \frac{1}{2}\ad_{C}\Bigr] &= \frac{q^2}{q!(2k\tilde g)^{q}k^{-1}N}\Phi_{q}=\frac{1}{\lambda_q}\Phi_q.
    \end{align*}
    The call to \textsf{LCSample} has $q$ sequential steps. With direct
    candidate-set updates, each step scans the Pauli-term list and performs
    polynomial-time support and Pauli-string operations. Counting the
    multiplicities $p_v$, sampling a permutation in $\Sigma_{\vec v}$, and
    evaluating $c_{\sigma,q}$ also take polynomial time in the input size and
    $q$. This proves the classical running-time bound.
\end{proof}

\subsection{Sampling commutators within the light-cone}
\label{sec:lcsample}
Algorithm~\ref{algo:LC_sampling} uses locality and rejection sampling to
obtain the required commutator distribution.

\begin{algorithm}[htbp]
\caption{Light-cone Commutator Sampling (\textsf{LCSample})}
\label{algo:LC_sampling}
\KwInput{Commutator-sampling data $\mathsf{CommData}$, order $q$.}
\KwOutput{Vector $\vec{v}\in[M]^{q}$ and operator $C$.}

Calculate total weight $W_{\text{tot}} \gets \sum_{v=1}^M\sum_{\gamma} |\alpha_{v, \gamma}|$ \;
Sample a tuple of indices $(v_1,\gamma_1)$ with probability $|\alpha_{v_1, \gamma_1}|/W_{\text{tot}}$\;
With probability $W_{\text{tot}}/(N \tilde g)$, \textbf{proceed}; otherwise \Return $((1,\ldots, 1),0)$\;

Initialize vector $\vec{v} \gets (v_1)$ and operator $C \gets \i\sgn(\alpha_{v_1,\gamma_1})P_{v_1, \gamma_1}$\;
Initialize candidates $\mathcal{A} \gets \{(v,\gamma) : \supp(P_{v,\gamma}) \cap \supp(P_{v_1,\gamma_1}) \neq \emptyset\}$\;

\For{$j \gets 2$ \textbf{to} $q$}{
    Calculate local weight $W_{\mathcal{A}} \gets \sum_{(v,\gamma) \in \mathcal{A}} |\alpha_{v, \gamma}|$\;
    Sample $(v_j,\gamma_j) \in \mathcal{A}$ with probability $|\alpha_{v_j, \gamma_j}|/W_{\mathcal{A}}$\;
    With probability $1-W_{\mathcal{A}}/(jk\tilde g)$ \Return $((1,\ldots, 1),0)$\; 
    \eIf{$[\i\sgn(\alpha_{v_j,\gamma_j})P_{v_j, \gamma_j}, C] \neq 0$}{
        $C \gets [\i\sgn(\alpha_{v_j,\gamma_j})P_{v_j, \gamma_j}, C]/2$\;
    }{
        \Return $((1,\ldots, 1),0)$\;
    }
    Append $v_j$ to $\vec{v}$\;
    $\mathcal{A} \gets \mathcal{A} \cup \{(v,\gamma): \supp(P_{v,\gamma}) \cap \supp(P_{v_j, \gamma_j}) \neq \emptyset\}$\;
}
\Return $(\vec{v},  C)$\;
\end{algorithm}

\begin{lemma}
\label{lem:lcsample}
    Under the same assumptions as in \lem{BCH-sample}, $\textsf{LCSample}(\mathsf{CommData},q)$ outputs an operator \begin{align*}
        C= \frac{[\i\sgn(\alpha_{v_q,\gamma_q})P_{v_q, \gamma_q}, \ldots, \i\sgn(\alpha_{v_1,\gamma_1})P_{v_1, \gamma_1}]}{2^{q-1}}
    \end{align*}
    in $\pm\i \mathcal{P}_N \cup \{0\}$ and the index vector \begin{align*}
        \vec{v}=\begin{cases}
         (v_1, \ldots, v_q), & \text{if } C\in \pm\i \mathcal{P}_N, \\
        (1,\ldots, 1), &\text{if } C=0,
        \end{cases}
    \end{align*} with probability \begin{align*}
        \frac{\|[\i\alpha_{v_q, \gamma_q} P_{v_q,\gamma_q}, \ldots, \i \alpha_{v_1, \gamma_1}P_{v_1,\gamma_1}]\|}{ q!(2k\tilde g)^{q-1}N\tilde g}.
    \end{align*}
\end{lemma}
\begin{proof}
For two Pauli operators $P,Q$, $\|[P, Q]\|$ equals $ 2\|P\|\|Q\|$ if they do not commute, and equals $0$ otherwise; hence $\|[\i\sgn(\alpha_{v_q,\gamma_q})P_{v_q,\gamma_q}, \ldots, \i\sgn(\alpha_{v_1,\gamma_1})P_{v_1,\gamma_1}]\| = 2^{q-1}$ if it is non-zero. Therefore, the probability stated in the lemma is
\begin{align*}
    &\frac{\|[\i\alpha_{v_q, \gamma_q} P_{v_q,\gamma_q}, \ldots, \i \alpha_{v_1, \gamma_1}P_{v_1,\gamma_1}]\|}{ q!(2k\tilde g)^{q-1}N\tilde g}\\
    ={}&\begin{cases}
        0 & \text{ if }\exists j\in[q], [P_{v_j, \gamma_j}, \ldots, P_{v_1,\gamma_1}] = 0 \\
         \frac{\prod_{j=1}^q  |\alpha_{v_j, \gamma_j}|}{q!(k\tilde g)^{q-1}N\tilde g} & \text{ else }.
    \end{cases}
\end{align*}
We verify the marginal probability produced by the rejection steps. 
In the first step, $v_1$ is chosen with probability $|\alpha_{v_1, \gamma_1}| / W_{\text{tot}}$ and accepted with probability $W_{\text{tot}} / (N \tilde g)$. The final probability is
\begin{align*}
    p((v_1,\gamma_1)) = \frac{|\alpha_{v_1, \gamma_1}|}{W_{\text{tot}}} \cdot \frac{W_{\text{tot}}}{N \tilde g} = \frac{|\alpha_{v_1, \gamma_1}|}{N \tilde g}.
\end{align*}
This step is valid because $W_{\text{tot}} \le N\tilde g $, ensuring the acceptance probability is at most 1.

In step $j$, the term $v_j$ is sampled from the candidate set $\mathcal{A}$ with probability $|\alpha_{v_j, \gamma_j}| / W_{\mathcal{A}}$ and accepted with probability $W_{\mathcal{A}} / (jk\tilde g)$. The conditional probability is
\begin{align*}
    p((v_j,\gamma_j) \mid (v_{1},\gamma_1),\ldots, (v_{j-1},\gamma_{j-1})) &= \frac{|\alpha_{v_j, \gamma_j}|}{W_{\mathcal{A}}} \cdot \frac{W_{\mathcal{A}}}{jk\tilde g} \\
    &= \frac{|\alpha_{v_j, \gamma_j}|}{jk\tilde g}.
\end{align*}
The candidate set is supported on the union of the first $j-1$ sampled
supports, whose size is at most $k+(j-2)(k-1)$. Hence
$W_{\mathcal A}\le (k+(j-2)(k-1))\tilde g\le jk\tilde g$, so the acceptance
probability is valid.

Multiplying these probabilities, the sequence $((v_1,\gamma_1), \ldots, (v_q,\gamma_q))$ is generated with probability
\begin{align*}
    &p((v_1,\gamma_1), \ldots, (v_q,\gamma_q))\\
    ={}& p((v_1,\gamma_1)) \prod_{j=2}^q p((v_j,\gamma_j) \mid (v_{1},\gamma_1),\ldots, (v_{j-1},\gamma_{j-1})) \\
    ={}& \frac{|\alpha_{v_1, \gamma_1}|}{N \tilde g} \prod_{j=2}^q \frac{|\alpha_{v_j, \gamma_j}|}{jk\tilde g} \\
    ={}&  \frac{\prod_{j=1}^q  |\alpha_{v_j, \gamma_j}|}{q!(k\tilde g)^{q-1}N\tilde g}.
\end{align*}
This is the probability claimed in the lemma.
    
\end{proof}

\subsection{Connected-cluster preprocessing}
\label{sec:connected-cluster-preprocessing}

For finite systems, we explicitly construct the Pauli decomposition of the
truncated BCH generator and merge identical Pauli contributions before the
Taylor expansion used in the numerical estimates. This preprocessing is more
expensive than the light-cone sampler used in the asymptotic analysis. The
bound below concerns only this preprocessing. The subsequent Taylor expansion
and conversion to an LCQC are performed separately.

We adapt the standard linked-cluster subtraction to the truncated BCH
generator; see, e.g., Ref.~\cite{TANG2013557}. Throughout this subsection, if
an operator-valued polynomial has the Pauli expansion
\begin{align}
\label{eq:pauli-table}
    A(x)
    =
    \sum_{q=0}^{q_0}
    \sum_{P\in\mathcal P_N}
    a_{q,P}x^qP,
\end{align}
we call the sparse coefficient array \(\{a_{q,P}\}_{q,P}\), with only its
nonzero entries stored, the \emph{Pauli table} of \(A(x)\). Addition and
subtraction of Pauli tables are performed coefficientwise, so entries with
the same degree \(q\) and Pauli string \(P\) are merged by adding their
coefficients. For a formal power series \(F(x)=\sum_{q\ge0}F_qx^q\), we write
\([F(x)]_{\le q_0}:=\sum_{q=0}^{q_0}F_qx^q\).

Let
\begin{align*}
    \mathfrak I
    =
    \{(v,\gamma):
    \tilde h_{v,\gamma}
    \text{ appears in }\eq{decompose-tilde-H}\}.
\end{align*}
The same Hamiltonian Pauli term is treated as a distinct element of
\(\mathfrak I\) when it appears in different product-formula stages. Define
the overlap graph \(G_{\rm ov}=(\mathfrak I,E)\) by placing an edge between
distinct vertices \((v,\gamma)\) and \((v',\gamma')\) whenever
\begin{align*}
    \supp(\tilde h_{v,\gamma})
    \cap
    \supp(\tilde h_{v',\gamma'})
    \ne\emptyset.
\end{align*}
A cluster is a subset \(C\subseteq\mathfrak I\), and it is connected if the
induced subgraph \(G_{\rm ov}[C]\) is connected.
For \(S\subseteq\mathfrak I\), let
\begin{align*}
    V_S(x)
    =
    \prod_{v=1}^{M}
    \exp\biggl(
        \sum_{\gamma:(v,\gamma)\in S}
        -\i\tilde h_{v,\gamma}
    \biggr),
\end{align*}
with the product order of \eq{def-VK-BCH}. Stages containing no term from
\(S\) contribute the identity. Define
\begin{align*}
    &\mc V_S(x)(\rho)
    =
    V_S(x)\rho V_S^\dagger(x),
    \\
    &\Phi(S;x)
    :=
    \log\mc V_S(x)
    =
    \sum_{q\ge1}\Phi_q(S;x).
\end{align*}
Since \(\mc V_S(x)=\exp(\ad_{\log V_S(x)})\), we have, as formal power
series,
\begin{align}
\label{eq:operator-channel-log}
    \Phi(S;x)
    =
    \ad_{\log V_S(x)}.
\end{align}
By construction, \(\Phi(\mathfrak I;x)=\Phi(x)\), and hence
\(\Phi_q(\mathfrak I;x)=\Phi_q(x)\) for every \(q\ge1\).
Define
\begin{align}
\label{eq:cluster-weight-recursion}
    \mc W(\emptyset;x)
    &=
    0,
    \nonumber\\
    \mc W(S;x)
    &=
    \Phi(S;x)
    -
    \sum_{T\subsetneq S}\mc W(T;x),
    \qquad S\ne\emptyset.
\end{align}
Equivalently,
\begin{align}
\label{eq:cluster-inversion-forward}
    \Phi(S;x)
    =
    \sum_{T\subseteq S}\mc W(T;x).
\end{align}

\begin{algorithm}[H]
\caption{Connected-cluster preprocessing for the BCH generator}
\label{algo:cluster-preprocessing}
\KwInput{Stage Pauli terms
\(\{\tilde h_{v,\gamma}\}_{v,\gamma}\) from
\eq{decompose-tilde-H} and BCH truncation order \(q_0\).}
\KwOutput{Merged Pauli table representing an operator \(G(x)\) satisfying
\(\ad_{G(x)}=\sum_{q=K+1}^{q_0}\Phi_q(x)\).}
Set
\(\mathfrak I\gets
\{(v,\gamma):
\tilde h_{v,\gamma}\text{ appears in }\eq{decompose-tilde-H}\}\)\;
Build the overlap graph \(G_{\rm ov}\) on \(\mathfrak I\)\;
Initialize an empty output Pauli table\;
\For{\(r=1,\ldots,q_0\)}{
    \ForEach{connected cluster \(C\subseteq\mathfrak I\) with \(|C|=r\)}{
        Compute the Pauli table of
        \([\log V_C(x)]_{\le q_0}\)\;
        Subtract the stored tables of all connected proper subclusters
        \(T\subsetneq C\)\;
        Merge equal Pauli strings in the resulting table\;
        Store the resulting table under \(C\)\;
        Add the resulting table to the output Pauli table\;
    }
}
\Return the output Pauli table\;
\end{algorithm}

For each connected cluster \(C\), we compute
\([\log V_C(x)]_{\le q_0}\) as a formal power series in \(x\):
\begin{align*}
    [V_C(x)]_{\le q_0}
    &=
    \biggl[
    \prod_{v=1}^{M}
    \sum_{p=0}^{q_0}
    \frac{1}{p!}
    \biggl(
        \sum_{\gamma:(v,\gamma)\in C}
        -\i\tilde h_{v,\gamma}
    \biggr)^p
    \biggr]_{\le q_0},
    \\
    [\log V_C(x)]_{\le q_0}
    &=
    \biggl[
    \sum_{n=1}^{q_0}
    \frac{(-1)^{n+1}}{n}
    \bigl([V_C(x)]_{\le q_0}-I\bigr)^n
    \biggr]_{\le q_0}.
\end{align*}
All multiplications are performed on Pauli tables and truncated at degree
\(q_0\), with equal Pauli strings merged after each multiplication. The
resulting Pauli strings are supported on
\(\bigcup_{(v,\gamma)\in C}\supp(\tilde h_{v,\gamma})\).

\begin{proposition}
\label{prop:connected-cluster-preprocessing}
Let \(\mathfrak I\) be as above, and suppose that each
\(\tilde h_{v,\gamma}\) is supported on at most \(k\) qubits. Then the table
returned by \algo{cluster-preprocessing} represents an operator \(G(x)\)
satisfying
\begin{align*}
    \ad_{G(x)}
    =
    \sum_{q=K+1}^{q_0}\Phi_q(x).
\end{align*}
Let \(\Delta_{\rm ov}\ge1\) be an upper bound on the maximum degree of
\(G_{\rm ov}\). The number of classical arithmetic operations is at most
\begin{align*}
    |\mathfrak I|\,
    e^{\mathcal O(q_0\log(\Delta_{\rm ov}+1)+kq_0)}
    \poly(q_0).
\end{align*}
\end{proposition}

\begin{proof}
We first prove by induction on \(|S|\) that \(\mc W(S;x)=0\) whenever \(S\)
is disconnected. Suppose that \(S\) has connected components
\(S^{(1)},\ldots,S^{(r)}\), where \(r\ge2\). Since the terms belonging to
different components have disjoint supports, their channels commute and may
be regrouped as
\begin{align*}
    \mc V_S(x)
    =
    \prod_{\alpha=1}^r\mc V_{S^{(\alpha)}}(x).
\end{align*}
Taking the formal logarithm gives
\begin{align*}
    \Phi(S;x)
    =
    \sum_{\alpha=1}^r\Phi(S^{(\alpha)};x).
\end{align*}
By the induction hypothesis, every disconnected proper subset of \(S\) has
zero weight, while every connected subset of \(S\) is contained in one of
the components. Therefore,
\begin{align*}
    \sum_{T\subsetneq S}\mc W(T;x)
    &=
    \sum_{\alpha=1}^r
    \sum_{T\subseteq S^{(\alpha)}}\mc W(T;x)
    \\
    &=
    \sum_{\alpha=1}^r\Phi(S^{(\alpha)};x)
    =
    \Phi(S;x),
\end{align*}
where the second equality follows from
\eq{cluster-inversion-forward}. Substituting this equality into
\eq{cluster-weight-recursion} gives \(\mc W(S;x)=0\).

We next show that every monomial in \(\mc W(S;x)\) contains at least one
factor associated with each vertex of \(S\). For each
\(a=(v,\gamma)\in S\), replace \(\tilde h_{v,\gamma}\) by
\(z_a\tilde h_{v,\gamma}\) in the definition of \(V_T(x)\) for every
\(T\subseteq S\). Denote the resulting quantities by
\(\Phi(T;x,\bm z)\) and \(\mc W(T;x,\bm z)\), with
\(\mc W(T;x,\bm z)\) defined by the recursion
\eq{cluster-weight-recursion}.

We prove the claim by induction on \(|S|\). Fix a monomial involving exactly
the variables indexed by \(R\subsetneq S\). Setting \(z_a=0\) for every
\(a\in S\setminus R\) reduces \(V_S(x)\) to \(V_R(x)\), so the coefficient
of the fixed monomial in \(\Phi(S;x,\bm z)\) is the same as its coefficient
in \(\Phi(R;x,\bm z)\). By the induction hypothesis, every monomial in
\(\mc W(T;x,\bm z)\) contains all variables indexed by \(T\). Hence the
fixed monomial can occur in \(\mc W(T;x,\bm z)\) only if \(T\subseteq R\).

Taking the coefficient of the fixed monomial in
\eq{cluster-weight-recursion}, its coefficient in
\(\mc W(S;x,\bm z)\) equals its coefficient in \(\Phi(S;x,\bm z)\) minus
the sum of its coefficients in \(\mc W(T;x,\bm z)\) over
\(T\subsetneq S\). The first coefficient is the same as in
\(\Phi(R;x,\bm z)\), and all terms with \(T\not\subseteq R\) vanish by the
induction hypothesis. Since \(R\subsetneq S\), the remaining coefficient is
therefore the coefficient of the fixed monomial in
\[
    \Phi(R;x,\bm z)
    -
    \sum_{T\subseteq R}\mc W(T;x,\bm z),
\]
which is zero by \eq{cluster-inversion-forward} applied to \(R\). Thus every
monomial in \(\mc W(S;x)\) contains at least one factor associated with each
vertex of \(S\). Since each such factor is linear in \(x\),
\[
    [\mc W(S;x)]_{\le q_0}
    =
    0
    \qquad
    \text{if }|S|>q_0.
\]

Together with the vanishing of disconnected cluster weights,
\eq{cluster-inversion-forward} and
\(\Phi(\mathfrak I;x)=\Phi(x)\) give
\begin{align}
\label{eq:connected-cluster-sum}
    [\Phi(x)]_{\le q_0}
    =
    \sum_{\substack{
        C\subseteq\mathfrak I\\
        C\ {\rm connected},\ |C|\le q_0
    }}
    [\mc W(C;x)]_{\le q_0}.
\end{align}

By induction on \(|C|\), the subtraction step in
\algo{cluster-preprocessing}, together with
\eq{cluster-weight-recursion} and \eq{operator-channel-log}, shows that the
table stored under \(C\) represents an operator \(G_C(x)\) satisfying
\begin{align*}
    \ad_{G_C(x)}
    =
    [\mc W(C;x)]_{\le q_0}.
\end{align*}
Summing the stored tables and applying
\eq{connected-cluster-sum} proves correctness.

For the cost, fix a root vertex. Every connected cluster of size \(r\)
containing the root has a spanning tree. A depth-first traversal of the tree
uses \(r-1\) steps that visit a new vertex and \(r-1\) backtracking steps.
The positions of the former can be chosen in at most
\(\binom{2r-2}{r-1}\le4^{r-1}\) ways, and each such step has at most
\(\Delta_{\rm ov}\) choices. Hence the number of connected clusters of size
\(r\) containing the root is at most
\((4\Delta_{\rm ov})^{r-1}\). The total number of processed clusters is at
most
\begin{align*}
    |\mathfrak I|
    \sum_{r=1}^{q_0}
    (4\Delta_{\rm ov})^{r-1}
    \le
    |\mathfrak I|q_0(4\Delta_{\rm ov})^{q_0-1}.
\end{align*}

For \(|C|\le q_0\), the union of the supports in \(C\) contains at most
\(kq_0\) qubits, so the Pauli table at each degree contains at most
\(4^{kq_0}\) entries. At most \(q_0\) stage exponentials are nontrivial.
Forming the truncated product and logarithm takes
\(\poly(q_0)\) table multiplications, each requiring
\(4^{\mathcal O(kq_0)}\) arithmetic operations. Subtracting the tables of all
connected proper subclusters requires at most \(2^{q_0}\) table
subtractions. The total cost is therefore
\begin{align*}
    &|\mathfrak I|q_0(4\Delta_{\rm ov})^{q_0-1}
    2^{q_0}4^{\mathcal O(kq_0)}\poly(q_0)
    \\
    ={}&
    |\mathfrak I|\,
    e^{\mathcal O(q_0\log(\Delta_{\rm ov}+1)+kq_0)}
    \poly(q_0).
\end{align*}
\end{proof}

For a fixed product formula,
\(|\mathfrak I|\le(\kappa_K+1)\Gamma\). If \(k\) is fixed and
\(\Delta_{\rm ov}=\mathcal O(1)\), the preprocessing cost is
\(\mathcal O(\Gamma)e^{\mathcal O(q_0)}\poly(q_0)\). It is therefore linear
in \(\Gamma\) for fixed \(q_0\), as in the resource estimates of
\sec{numerics}. Under the same assumptions, if
\(q_0=\mathcal O(\log(1/\eps))\), the preprocessing cost is polynomial in
\(1/\eps\).


\section{Numerical experiments}
\label{sec:numerics}

\begin{figure*}[t!]
    \centering
    \includegraphics[width=\linewidth]
    {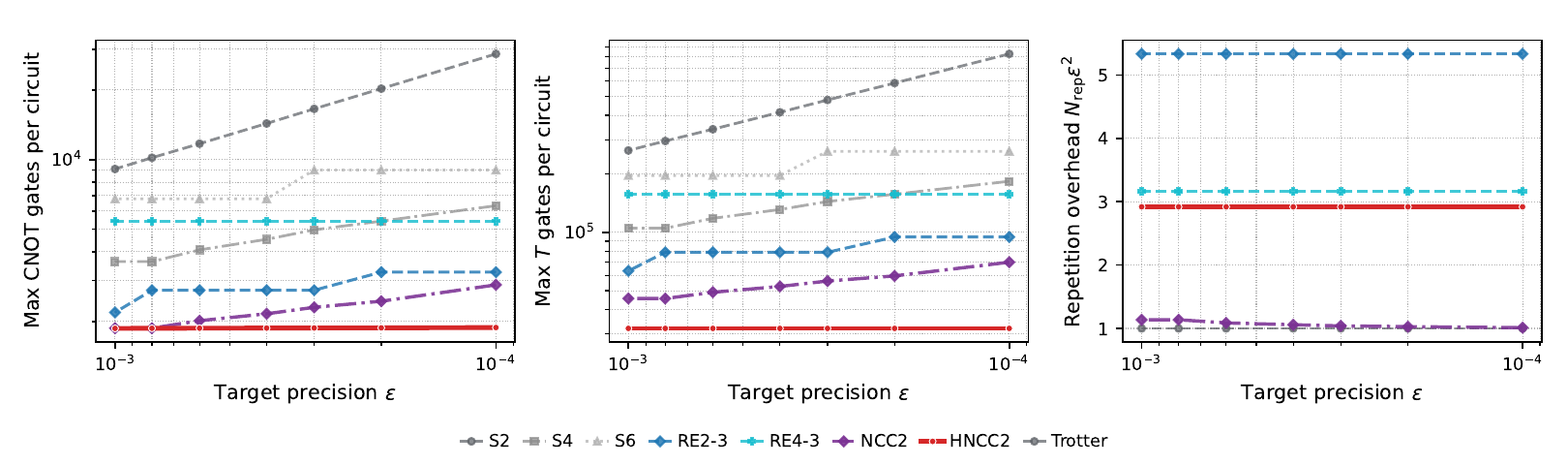}
    \caption{Finite-size resource comparison for the periodic Heisenberg
    chain with \(N=10\) and \(t=1\). The three panels show the maximum CNOT
    count per circuit, the estimated \(T\)-gate count per circuit, and the
    normalized repetition overhead \(N_{\rm rep}\eps^2\), respectively. S2,
    S4, and S6 denote the second-, fourth-, and sixth-order Suzuki formulas.
    RE2-3 and RE4-3 denote three-point Richardson extrapolation based on S2
    and S4, respectively. NCC2 and HNCC2 denote NCC and HNCC applied to S2.
    For HNCC2, we set \(q_0=8\) and \(m=5\).}
    \label{fig:resource-eps}
\end{figure*}

\begin{figure*}[t!]
    \centering
    \includegraphics[width=\linewidth]
    {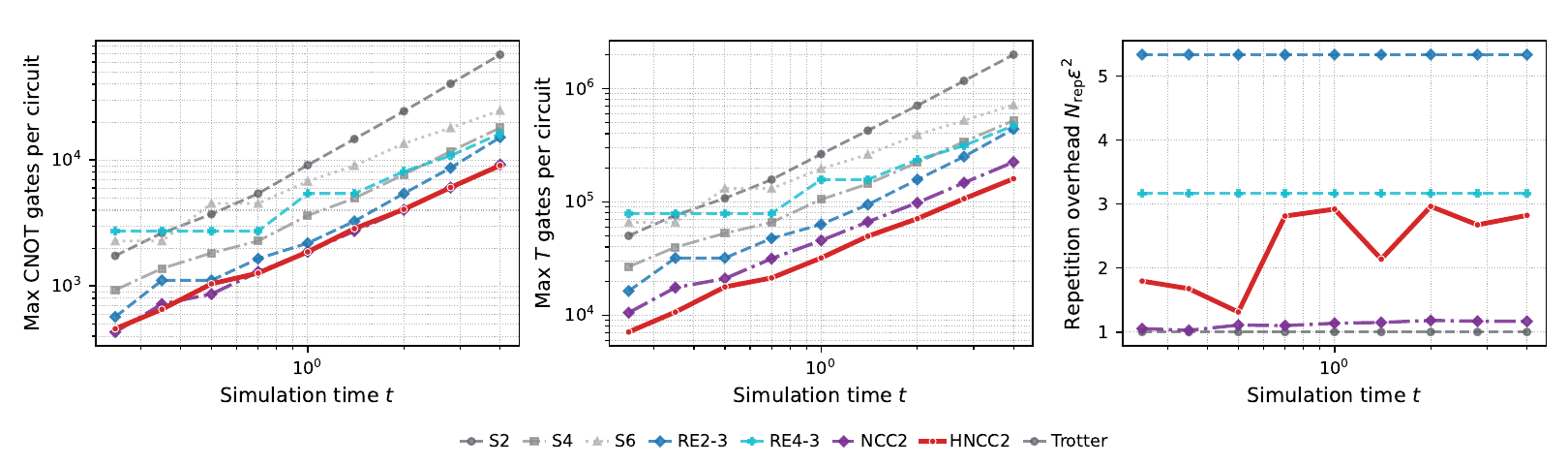}
    \caption{Finite-size resource comparison as a function of the simulation
    time \(t\), with target precision \(\eps=10^{-3}\) and \(N=10\). The
    panels show the maximum CNOT count per circuit, the estimated \(T\)-gate
    count per circuit, and the normalized repetition overhead
    \(N_{\rm rep}\eps^2\), using the conventions of \fig{resource-eps}. For
    HNCC2, we set \(q_0=8\) and \(m=5\).}
    \label{fig:resource-time}
\end{figure*}

The finite-size comparison includes Suzuki--Trotter formulas, Richardson
extrapolation, NCC, and HNCC on the periodic Heisenberg chain.

\subsection{Numerical setup}

The benchmark Hamiltonian is the periodic one-dimensional Heisenberg model
\begin{align*}
    H
    ={}&
    J\sum_{j=1}^N
    (X_jX_{j+1}+Y_jY_{j+1}+Z_jZ_{j+1})\\
    &+h\sum_{j=1}^N Z_j,
    \qquad N+1\equiv1.
\end{align*}
We set \(J=h=1\). For even \(N\), we partition the Pauli terms in \(H\) into
three groups such that all terms within each group commute:
\begin{align*}
    A&=J\sum_{\substack{1\le j\le N\\j\ {\rm odd}}}
    (X_jX_{j+1}+Y_jY_{j+1}+Z_jZ_{j+1}),\\
    B&=J\sum_{\substack{1\le j\le N\\j\ {\rm even}}}
    (X_jX_{j+1}+Y_jY_{j+1}+Z_jZ_{j+1}),\\
    C&=h\sum_{j=1}^N Z_j,
    \qquad H=A+B+C.
\end{align*}

For HNCC, we construct the truncated BCH generator using the connected-cluster
preprocessing of \sec{connected-cluster-preprocessing}, merge equal Pauli
strings, and convert the resulting Pauli table into an LCQC using the
parameter-shift and identity-pairing identities.

Errors are evaluated assuming exact Pauli rotations; rotation-synthesis error
is discussed separately below.
Errors are measured by the trace norm over computational-basis inputs. For
\(z\in\{0,1\}^N\), let
\(\rho_z=\mc U(t)(\ket{z}\bra{z})\), and let \(\sigma_{z,\nu}\) denote the
output produced by the method under consideration using \(\nu\) Trotter
steps. All methods are required to satisfy
\[
    \max_{z\in\{0,1\}^N}
    \bigl\|
        \sigma_{z,\nu}-\rho_z
    \bigr\|_1
    \le\eps.
\]
For HNCC, we set \(q_0=8\) and require
\(\lambda_{\mathrm{paired}}^{m\nu}\le1.75\), where
\(\lambda_{\mathrm{paired}}\) is the LCQC \(1\)-norm of one fractional
Trotter remainder at step size \(x=t/\nu\). At each parameter point, we choose the smallest
\(\nu\) for which this constraint holds and a global cutoff can make the
accumulated BCH truncation and cutoff errors sum to at most \(\eps\). We then
take the smallest such cutoff \(s_0\). At \(t=1\) and \(\eps=10^{-3}\),
minimizing the sum of the maximum CNOT and estimated \(T\)-gate counts over
$m\in\{1,\ldots,10\}$ gives \(m=5\), which is used in both scans.

\subsection{Resource comparison}

We compare HNCC with second-, fourth-, and sixth-order Suzuki--Trotter
formulas, NCC~\cite{zeng2025simple}, and Richardson
extrapolation~\cite{watson2025exponentially} using the optimized points from
Ref.~\cite{low2019wellconditionedmultiproducthamiltoniansimulation}.

Gate counts include CNOT gates and non-Clifford \(R_z\) rotations that require
synthesis, but omit single-qubit Clifford gates. A Pauli rotation
\(e^{-\i\vartheta P}\) generated by
a Pauli string with support size \(w\) is counted as \(2(w-1)\) CNOT gates and
one \(R_z\) rotation. A controlled Pauli rotation with support size $w$ used
by NCC is decomposed into $2w$ CNOT gates and two
$R_z$ rotations. The connected-cluster
preprocessing returns an operator \(G(x)\) satisfying
\(\ad_{G(x)}=\sum_{q=K+1}^{q_0}\Phi_q(x)\). Writing its Pauli expansion in
the form of \eq{pauli-table} as
\begin{align*}
    G(x)
    =
    \sum_{q=K+1}^{q_0}x^q
    \sum_{P\in\mathcal P_N}g_{q,P}P,
\end{align*}
define
\[
    C_{\mathrm{CNOT}}(q)
    :=
    \max_{P:g_{q,P}\ne0}
    2\bigl(|\supp(P)|-1\bigr).
\]
Each fractional Trotter remainder uses at most one Pauli rotation from the
linear part. Its angle has the common magnitude
$\theta=\tan^{-1}(\lambda_{\mathrm{single}}/2)$. Products of multiple BCH terms use only \(\pi/4\) Pauli
rotations, which are Clifford gates. Their CNOT cost is included below, but
they require no rotation synthesis. Define
\begin{align*}
    C_{\rm multi}(s_0)
    :={}&
    \max\biggl\{0,
    \max_{\substack{r\ge2,\ q_\ell\in[K+1,q_0]\\
    \sum_{\ell=1}^r q_\ell\le s_0}}
    \sum_{\ell=1}^r C_{\mathrm{CNOT}}(q_\ell)
    \biggr\}.
\end{align*}
The maximum gate counts for HNCC with parameters \((m,\nu_m,s_0)\) are
\begin{align*}
    N_{\mathrm{CNOT}}^{\max}
    &={}
    \nu_m
    \biggl[
        12N
        +
        m\max_{K+1\le q\le q_0}C_{\mathrm{CNOT}}(q)
    \biggr]
    \\
    &\qquad
    +C_{\rm multi}(s_0),
    \\
    N_{R_z}^{\max}
    &=
    \nu_m(7N+m).
\end{align*}

To obtain an estimated \(T\)-gate count, we synthesize each non-Clifford
\(R_z\) rotation over the Clifford+\(T\) gate set. We use the average
synthesis cost
\(c_T(\eps_{\rm syn})
\simeq1.149\log_2(1/\eps_{\rm syn})+9.2\)
for a random rotation angle~\citep{bocharov2015efficient}, giving
\(N_T^{\rm est}=c_T(\eps_{\rm syn})N_{R_z}^{\max}\). We take
\(\eps_{\rm syn}=10^{-10}\), which gives \(c_T\simeq47\). Even if the
synthesis errors accumulate linearly, the resulting per-circuit error bound
is below \(9\times10^{-6}\) for all data points shown.

We also report the normalized repetition overhead \(N_{\rm rep}\eps^2\),
omitting the common Hoeffding factor \(2\log(2/\delta)\). We use
\(N_{\rm rep}=\eps^{-2}\) for Suzuki--Trotter formulas and
\(N_{\rm rep}=r\|\alpha\|_1^2\eps^{-2}\) for an \(r\)-point extrapolation
formula with coefficients \(\alpha\). For NCC, we use
\(N_{\rm rep}=(\lambda_{\rm NCC}^{\nu})^4\eps^{-2}\), where
\(\lambda_{\rm NCC}\) is the coefficient \(1\)-norm of the one-step
unitary LCU representation of the compensation operator. For HNCC, we use
\(N_{\rm rep}=\lambda_{\mathrm{paired}}^{2m\nu}\eps^{-2}\).

As shown in \fig{resource-eps} and \fig{resource-time}, HNCC2 has CNOT counts
comparable to NCC2 and extrapolation while giving the lowest estimated
\(T\)-gate count per circuit over both parameter ranges. Compared with the
uncompensated \(S_2\) formula, HNCC2 achieves a \(4.9\)--\(15.2\times\)
reduction in the CNOT count and an \(8.3\)--\(25.9\times\) reduction in the
estimated \(T\)-gate count
over the range of target precisions considered.
For this precision range, the minimum number of HNCC steps is determined by
the constraint \(\lambda_{\mathrm{paired}}^{m\nu}\le1.75\), rather than by the
target precision. Consequently, \(\nu\) and the estimated \(T\)-gate count
remain unchanged, while increasing the global cutoff causes only a small
change in the CNOT count. Over the time range in
\fig{resource-time}, the corresponding reductions are
\(3.6\)--\(7.6\times\) in the CNOT count and \(6.1\)--\(12.4\times\) in the
estimated \(T\)-gate count.
Both scans use the fixed value \(m=5\).
The normalized HNCC repetition overhead is at most
\(1.75^2\approx3.06\).

\section{Discussion and outlook}

HNCC combines channel-level Trotter-error compensation with a truncated BCH
representation of the remainder. The resulting algorithm is ancilla-free,
preserves nested-commutator scaling, and achieves effective \(2K+1\)-order
time scaling with polylogarithmic precision dependence in the circuit size.

In the numerical experiments, we reduce the LCQC \(1\)-norm by merging
identical Pauli terms in the truncated BCH generator using connected-cluster
preprocessing. For fixed locality and bounded overlap degree, this procedure
has polynomial classical cost and is linear in the number of Pauli terms for
fixed \(q_0\). Whether the LCQC \(1\)-norm can be reduced efficiently for
general \(k\)-local Hamiltonians remains open.

The same channel-level compensation may also apply beyond Hamiltonian
simulation. Trotter error mitigation has recently been used in several
related settings, including eigenstate problems~\cite{sun2026high},
Lindbladian simulation~\cite{yu2025lindbladian,wang2026lindbladian}, and
quantum singular value transformation~\cite{chakraborty2025quantum}. Whether channel-level compensation extends to the product-formula
approximations in these algorithms remains open.

\paragraph{Note added.}
While finalizing this manuscript, we became aware of related independent work by Murota et al.~\cite{murota2026unbiased}, which introduces
probabilistic Trotter error reversal (PTER). PTER is an unbiased method based on a substantially different compensation
strategy, which reverses the coherent time-dependent dynamics generated by
the Trotter error. For geometrically local Hamiltonians with \(\Gamma=O(N)\) Pauli terms, it
achieves an expected Pauli-rotation count
\(O(N^{1+1/(2K+1)}t^{1+1/(2K+1)})\) with constant sampling overhead. HNCC
provides a complementary guarantee: a maximum elementary-gate bound for
general \(k\)-local Hamiltonians with an arbitrary number \(\Gamma\) of Pauli
terms.

\section*{Acknowledgements}
X.W., S.Z., Z.W., and T.L. are supported by the National Natural Science Foundation of China under Grant No.~62372006. T.G. is supported by ERC Starting Grant No.~101163189 and UKRI Future Leaders Fellowship MR/X023583/1. P.Z. is supported by Shanghai Baiyulan Pujiang Project (Grant No.~25PJA066).

\bibliographystyle{qubit-doi}
\bibliography{reference}

@article{sun2026high,
  title = {High-precision and low-depth quantum algorithm design for eigenstate problems},
  author = {Sun, Jinzhao and Zeng, Pei and Gur, Tom and Kim, MS},
  journal = {Science Advances},
  volume = {12},
  number = {3},
  pages = {eaeb1622},
  year = {2026},
  publisher = {American Association for the Advancement of Science},
  doi = {10.1126/sciadv.aeb1622}
}

@misc{low2019wellconditionedmultiproducthamiltoniansimulation,
      title={Well-conditioned multiproduct Hamiltonian simulation}, 
      author={Guang Hao Low and Vadym Kliuchnikov and Nathan Wiebe},
      year={2019},
      eprint={1907.11679},
      archivePrefix={arXiv},
      primaryClass={quant-ph},
      url={https://arxiv.org/abs/1907.11679}, 
}

@article{mcardle2020quantum,
  title = {Quantum computational chemistry},
  author = {McArdle, Sam and Endo, Suguru and Aspuru-Guzik, Al{\'a}n and Benjamin, Simon C and Yuan, Xiao},
  journal = {Reviews of Modern Physics},
  volume = {92},
  number = {1},
  pages = {015003},
  year = {2020},
  publisher = {American Physical Society},
  doi = {10.1103/RevModPhys.92.015003}
}

@article{zeng2025simple,
  title = {Simple and high-precision {H}amiltonian simulation by compensating {T}rotter error with linear
    combination of unitary operations},
  author = {Zeng, Pei and Sun, Jinzhao and Jiang, Liang and Zhao, Qi},
  journal = {PRX Quantum},
  volume = {6},
  number = {1},
  pages = {010359},
  year = {2025},
  publisher = {American Physical Society},
  doi = {10.1103/prxquantum.6.010359}
}

@article{bocharov2015efficient,
  title = {Efficient synthesis of universal repeat-until-success quantum circuits},
  author = {Bocharov, Alex and Roetteler, Martin and Svore, Krysta M.},
  journal = {Physical Review Letters},
  volume = {114},
  number = {8},
  pages = {080502},
  year = {2015},
  publisher = {American Physical Society},
  doi = {10.1103/PhysRevLett.114.080502}
}

@article{harrow2009quantum,
  title = {Quantum Algorithm for Linear Systems of Equations},
  author = {Harrow, Aram W. and Hassidim, Avinatan and Lloyd, Seth},
  journal = {Physical Review Letters},
  volume = {103},
  issue = {15},
  pages = {150502},
  numpages = {4},
  year = {2009},
  month = {Oct},
  publisher = {American Physical Society},
  doi = {10.1103/PhysRevLett.103.150502}
}

@article{childs2017quantum,
  author = {Childs, Andrew M. and Kothari, Robin and Somma, Rolando D.},
  title = {Quantum Algorithm for Systems of Linear Equations with Exponentially Improved Dependence on
    Precision},
  journal = {SIAM Journal on Computing},
  volume = {46},
  number = {6},
  pages = {1920-1950},
  year = {2017},
  doi = {10.1137/16M1087072}
}

@article{cho2024doubling,
  title = {Doubling the order of approximation via the randomized product formula},
  author = {Cho, Chien-Hung and Berry, Dominic W. and Hsieh, Min-Hsiu},
  journal = {Physical Review A},
  volume = {109},
  issue = {6},
  pages = {062431},
  numpages = {11},
  year = {2024},
  month = {Jun},
  publisher = {American Physical Society},
  doi = {10.1103/PhysRevA.109.062431}
}

@article{mizuta2026commutator,
  doi = {10.22331/q-2026-01-19-1974},
  title = {On the commutator scaling in {H}amiltonian simulation with multi-product formulas},
  author = {Mizuta, Kaoru},
  journal = {{Quantum}},
  issn = {2521-327X},
  publisher = {{Verein zur F{\"{o}}rderung des Open Access Publizierens in den Quantenwissenschaften}},
  volume = {10},
  pages = {1974},
  month = jan,
  year = {2026}
}

@misc{mineh2025improving,
  title = {Improving time dynamics simulation by sampling the error unitary},
  author = {Lana Mineh and Adrian Chapman and Raul A. Santos},
  year = {2025},
  eprint = {arXiv:2508.17542}
}

@article{rendon2024improved,
  doi = {10.22331/q-2024-02-26-1266},
  title = {Improved {A}ccuracy for {T}rotter {S}imulations {U}sing {C}hebyshev {I}nterpolation},
  author = {Rendon, Gumaro and Watkins, Jacob and Wiebe, Nathan},
  journal = {{Quantum}},
  issn = {2521-327X},
  publisher = {{Verein zur F{\"{o}}rderung des Open Access Publizierens in den Quantenwissenschaften}},
  volume = {8},
  pages = {1266},
  month = feb,
  year = {2024}
}

@article{endo2019mitigating,
  title = {Mitigating algorithmic errors in a {H}amiltonian simulation},
  author = {Endo, Suguru and Zhao, Qi and Li, Ying and Benjamin, Simon and Yuan, Xiao},
  journal = {Physical Review A},
  volume = {99},
  issue = {1},
  pages = {012334},
  numpages = {8},
  year = {2019},
  month = {Jan},
  publisher = {American Physical Society},
  doi = {10.1103/PhysRevA.99.012334}
}

@article{carrera2023wellconditioned,
  doi = {10.22331/q-2023-07-25-1067},
  title = {Well-conditioned multi-product formulas for hardware-friendly {H}amiltonian simulation},
  author = {Carrera Vazquez, Almudena and Egger, Daniel J. and Ochsner, David and Woerner, Stefan},
  journal = {{Quantum}},
  issn = {2521-327X},
  publisher = {{Verein zur F{\"{o}}rderung des Open Access Publizierens in den Quantenwissenschaften}},
  volume = {7},
  pages = {1067},
  month = jul,
  year = {2023}
}

@article{low2017optimal,
  title = {Optimal {H}amiltonian Simulation by Quantum Signal Processing},
  author = {Low, Guang Hao and Chuang, Isaac L.},
  journal = {Physical Review Letters},
  volume = {118},
  issue = {1},
  pages = {010501},
  numpages = {5},
  year = {2017},
  month = {Jan},
  publisher = {American Physical Society},
  doi = {10.1103/PhysRevLett.118.010501}
}

@inproceedings{berry2014exponential,
  author = {Berry, Dominic W. and Childs, Andrew M. and Cleve, Richard and Kothari, Robin and Somma, Rolando D.},
  title = {Exponential improvement in precision for simulating sparse {H}amiltonians},
  year = {2014},
  isbn = {9781450327107},
  booktitle = {Proceedings of the Forty-Sixth Annual ACM Symposium on Theory of Computing},
  pages = {283–292},
  numpages = {10},
  location = {New York, New York},
  series = {STOC '14},
  doi = {10.1145/2591796.2591854}
}

@inproceedings{berry2015hamiltonian,
  author = {Berry, Dominic W. and Childs, Andrew M. and Kothari, Robin},
  title = {{H}amiltonian Simulation with Nearly Optimal Dependence on all Parameters},
  year = {2015},
  isbn = {9781467381918},
  publisher = {IEEE Computer Society},
  address = {USA},
  doi = {10.1109/FOCS.2015.54},
  booktitle = {Proceedings of the 2015 IEEE 56th Annual Symposium on Foundations of Computer Science (FOCS)},
  pages = {792–809},
  numpages = {18},
  series = {FOCS '15}
}

@article{childs2018toward,
  author = {Andrew M. Childs and Dmitri Maslov and Yunseong Nam and Neil J. Ross and Yuan Su},
  title = {Toward the first quantum simulation with quantum speedup},
  journal = {Proceedings of the National Academy of Sciences},
  volume = {115},
  number = {38},
  pages = {9456-9461},
  year = {2018},
  doi = {10.1073/pnas.1801723115}
}

@misc{farhi2014quantum,
  title = {A quantum approximate optimization algorithm},
  author = {Farhi, Edward and Goldstone, Jeffrey and Gutmann, Sam},
  year = {2014},
  eprint = {arXiv:1411.4028}
}

@article{TANG2013557,
title = {A short introduction to numerical linked-cluster expansions},
journal = {Computer Physics Communications},
volume = {184},
number = {3},
pages = {557-564},
year = {2013},
issn = {0010-4655},
doi = {https://doi.org/10.1016/j.cpc.2012.10.008},
url = {https://www.sciencedirect.com/science/article/pii/S0010465512003414},
author = {Baoming Tang and Ehsan Khatami and Marcos Rigol}
}

@article{an2023linear,
  title = {Linear Combination of {H}amiltonian Simulation for Nonunitary Dynamics with Optimal State
    Preparation Cost},
  author = {An, Dong and Liu, Jin-Peng and Lin, Lin},
  journal = {Physical Review Letters},
  volume = {131},
  issue = {15},
  pages = {150603},
  numpages = {6},
  year = {2023},
  month = {Oct},
  publisher = {American Physical Society},
  doi = {10.1103/PhysRevLett.131.150603}
}

@article{dong2022ground,
  title = {Ground-State Preparation and Energy Estimation on Early Fault-Tolerant Quantum Computers via
    Quantum Eigenvalue Transformation of Unitary Matrices},
  author = {Dong, Yulong and Lin, Lin and Tong, Yu},
  journal = {PRX Quantum},
  volume = {3},
  issue = {4},
  pages = {040305},
  numpages = {25},
  year = {2022},
  month = {Oct},
  publisher = {American Physical Society},
  doi = {10.1103/PRXQuantum.3.040305}
}

@article{watson2025exponentially,
  title = {Exponentially reduced circuit depths using {T}rotter error mitigation},
  author = {Watson, James D. and Watkins, Jacob},
  journal = {PRX Quantum},
  volume = {6},
  number = {3},
  pages = {030325},
  year = {2025},
  publisher = {American Physical Society},
  doi = {10.1103/kw39-yxq5}
}

@article{wiebe2010higher,
  title = {Higher order decompositions of ordered operator exponentials},
  author = {Wiebe, Nathan and Berry, Dominic and H{\o}yer, Peter and Sanders, Barry C.},
  journal = {Journal of Physics A: Mathematical and Theoretical},
  volume = {43},
  number = {6},
  pages = {065203},
  year = {2010},
  publisher = {IOP Publishing},
  doi = {10.1088/1751-8113/43/6/065203}
}

@misc{aftab2024multi,
  title = {Multi-product {H}amiltonian simulation with explicit commutator scaling},
  author = {Aftab, Junaid and An, Dong and Trivisa, Konstantina},
  year = {2024},
  eprint = {arXiv:2403.08922}
}

@misc{chakraborty2025quantum,
  title = {Quantum singular value transformation without block encodings: Near-optimal complexity with
    minimal ancilla},
  author = {Chakraborty, Shantanav and Hazra, Soumyabrata and Li, Tongyang and Shao, Changpeng and Wang,
    Xinzhao and Zhang, Yuxin},
  year = {2025},
  eprint = {arXiv:2504.02385}
}

@article{yu2025lindbladian,
  title = {Lindbladian Simulation with Logarithmic Precision Scaling via Two Ancillas},
  author = {Yu, Wenjun and Li, Xiaogang and Zhao, Qi and Yuan, Xiao},
  journal = {Physical Review Letters},
  volume = {135},
  issue = {16},
  pages = {160602},
  numpages = {8},
  year = {2025},
  month = {Oct},
  publisher = {American Physical Society},
  doi = {10.1103/2cx4-b82c}
}

@misc{murota2026unbiased,
  title = {Unbiased {H}amiltonian Simulation by Reversing {T}rotter Error Dynamics},
  author = {Keisuke Murota and Yuta Kikuchi and Enrico Rinaldi and Frédéric Sauvage and Synge Todo},
  year = {2026},
  eprint = {arXiv:2606.29741}
}

@article{arnal2021note,
  title = {A note on the Baker--Campbell--Hausdorff series in terms of right-nested commutators},
  author = {Arnal, Ana and Casas, Fernando and Chiralt, Cristina},
  journal = {Mediterranean Journal of Mathematics},
  volume = {18},
  number = {2},
  pages = {53},
  year = {2021},
  publisher = {Springer},
  doi = {10.1007/s00009-020-01681-6}
}

@misc{wang2026lindbladian,
  title = {Lindbladian Simulation with Commutator Bounds},
  author = {Wang, Xinzhao and Zhou, Shuo and Wang, Xiaoyang and Zheng, Yi-Cong and Zhang, Shengyu and Li,
    Tongyang},
  year = {2026},
  eprint = {arXiv:2603.28602}
}

@article{feynman1982simulating,
  title = {Simulating physics with computers},
  author = {Feynman, Richard P.},
  journal = {International Journal of Theoretical Physics},
  volume = {21},
  number = {6/7},
  pages = {467--488},
  year = {1982},
  doi = {10.1007/BF02650179}
}

@article{lloyd1996universal,
  title = {Universal Quantum Simulators},
  author = {Lloyd, Seth},
  journal = {Science},
  volume = {273},
  number = {5278},
  pages = {1073--1078},
  year = {1996},
  doi = {10.1126/science.273.5278.1073}
}

@article{suzuki1991general,
  title = {General theory of fractal path integrals with applications to many-body theories and statistical
    physics},
  author = {Suzuki, Masuo},
  journal = {Journal of Mathematical Physics},
  volume = {32},
  number = {2},
  pages = {400--407},
  year = {1991},
  doi = {10.1063/1.529425}
}

@article{childs2021theory,
  title = {Theory of {T}rotter Error with Commutator Scaling},
  author = {Childs, Andrew M. and Su, Yuan and Tran, Minh C. and Wiebe, Nathan and Zhu, Shuchen},
  journal = {Physical Review X},
  volume = {11},
  number = {1},
  pages = {011020},
  year = {2021},
  doi = {10.1103/PhysRevX.11.011020}
}

@article{berry2015simulating,
  title = {Simulating {H}amiltonian dynamics with a truncated {T}aylor series},
  author = {Berry, Dominic W. and Childs, Andrew M. and Cleve, Richard and Kothari, Robin and Somma, Rolando D.},
  journal = {Physical Review Letters},
  volume = {114},
  number = {9},
  pages = {090502},
  year = {2015},
  doi = {10.1103/PhysRevLett.114.090502}
}

@article{low2019hamiltonian,
  title = {{H}amiltonian simulation by qubitization},
  author = {Low, Guang Hao and Chuang, Isaac L.},
  journal = {Quantum},
  volume = {3},
  pages = {163},
  year = {2019},
  doi = {10.22331/q-2019-07-12-163}
}
\end{document}